\documentclass[10pt,journal,compsoc]{IEEEtran}

\ifCLASSOPTIONcompsoc
  \usepackage[nocompress]{cite}
\else
  \usepackage{cite}
\fi
\ifCLASSINFOpdf
\else
\fi

\usepackage{amssymb}
\usepackage{amsmath}
\usepackage{amsthm}

\newtheorem{lemma}{Lemma}
\newtheorem{definition}{Definition}

\usepackage{multirow}
\usepackage{placeins}
\usepackage{makecell}
\usepackage{array}
\newcolumntype{L}[1]{>{\raggedright\let\newline\\\arraybackslash\hspace{0pt}}m{#1}} 
\newcolumntype{C}[1]{>{\centering\let\newline\\\arraybackslash\hspace{0pt}}m{#1}} 
\newcolumntype{R}[1]{>{\raggedleft\let\newline\\\arraybackslash\hspace{0pt}}m{#1}}
\def\barr{\begin{tabular}[c]{@{}c@{}}}
\def\earr{\end{tabular}}

\usepackage{pgfplots}
\usepgfplotslibrary{colorbrewer}
\usetikzlibrary{pgfplots.statistics,pgfplots.colorbrewer} 
\usepackage{pgfplotstable}
\usetikzlibrary{tikzmark}
\pgfplotsset{
    compat=1.9, 
    every tick label/.append style={font=\Large},
    label style={font=\Large}
}
\usepackage{graphicx,import}
\usepackage[subrefformat=parens,labelformat=parens]{subfig}
\usepackage{tikz}
\usetikzlibrary{positioning}
\usetikzlibrary{matrix}
\usepackage{graphicx,wrapfig,lipsum}

\usepackage{mathtools}
\DeclarePairedDelimiter\ceil{\lceil}{\rceil}
\DeclarePairedDelimiter\floor{\lfloor}{\rfloor}

\usepackage[linesnumbered,ruled]{algorithm2e} 
\usepackage{paralist} 
\usepackage{flushend} %
\usepackage[hyphens]{url} 

\makeatletter
\newcommand*{\@rowstyle}{}
\newcommand*{\rowstyle}[1]{
  \gdef\@rowstyle{#1}%
  \@rowstyle\ignorespaces%
}
\newcolumntype{=}{
  >{\gdef\@rowstyle{}}%
}
\newcolumntype{+}{
  >{\@rowstyle}%
}
\makeatother
\newcommand\myverb{}
\def\myverb|#1#2|{\underline{#1}#2}
\newcommand\U[1]{\underline{#1}}

\definecolor{ggreen}{HTML}{138a07}
\definecolor{bblue}{rgb}{0.06, 0.2, 0.8}

\newcommand{\nc}[1]{#1} 
\newcommand{\oc}[1]{#1} 

\newcommand{\dir}{Dir\_OA}
\newcommand{\itr}{It\_OA}
\newcommand{\gdir}{G\_DirA}
\newcommand{\gitr}{G\_ItA}
\newcommand{\ndir}{N\_DirA}
\newcommand{\nitr}{N\_ItA}

\newcommand{\Z}{\mathbb{Z}^{\geq 0}}

\definecolor{vviolet}{rgb}{0.5, 0.0, 1.0}
\definecolor{ultramarineblue}{rgb}{0.25, 0.4, 0.96}
\definecolor{hotmagenta}{rgb}{1.0, 0.11, 0.81}
\definecolor{cadmiumgreen}{rgb}{0.0, 0.42, 0.24}
\definecolor{kellygreen}{rgb}{0.3, 0.73, 0.09}
\definecolor{sand}{rgb}{0.76, 0.7, 0.5}
\definecolor{smalt(darkpowderblue)}{rgb}{0.0, 0.2, 0.6}
\definecolor{unitednationsblue}{rgb}{0.36, 0.57, 0.9}
\pgfplotscreateplotcyclelist{my black white}{%
solid,black, every mark/.append style={solid, fill=white},mark=pentagon\\%
solid,blue, every mark/.append style={solid, fill=white},mark=square*\\%
dashdotted,red, every mark/.append style={solid, fill=white}, mark=triangle\\%
solid,brown, every mark/.append style={solid, fill=white}, mark=*\\%
loosely dashed,black, every mark/.append style={solid, fill=white}, mark=diamond\\%
solid,green!80!black, very thick, every mark/.append style={solid, fill=white}, mark=star\\%
densely dotted,violet, every mark/.append style={solid, fill=white}, mark=otimes\\%
solid,blue(ryb), every mark/.append style={solid, fill=white}, mark=pentagon\\%
densely dotted,rred, every mark/.append style={solid, fill=white}, mark=asterisk\\%
solid, every mark/.append style={solid, fill=gray}, mark=triangle*\\%
loosely dashed, thick, every mark/.append style={solid, fill=gray},mark=*\\%
solid, every mark/.append style={solid, fill=gray},mark=square*\\%
dashdotted, every mark/.append style={solid, fill=gray},mark=otimes*\\%
dashdotdotted, every mark/.append style={solid},mark=star\\%
densely dashdotted,every mark/.append style={solid, fill=gray},mark=diamond*\\%
}
\pgfplotscreateplotcyclelist{my six colors}{%
solid,magenta, every mark/.append style={solid, fill=magenta}, mark=triangle\\%
dashdotted,red, every mark/.append style={solid, fill=white}, mark=star\\%
solid,cadmiumgreen, very thick, every mark/.append style={solid, fill=cadmiumgreen}, mark=otimes\\%
dashdotted,kellygreen, every mark/.append style={solid, fill=kellygreen}, mark=*\\%
solid,blue!50!red, every mark/.append style={solid, fill=white},mark=square*\\%
dashdotted,blue, every mark/.append style={solid, fill=white}, mark=diamond\\%
}
\pgfplotscreateplotcyclelist{my three colors}{%
dashdotted,red, every mark/.append style={solid, fill=white}, mark=star\\%
dashdotted,kellygreen, every mark/.append style={solid, fill=kellygreen}, mark=*\\%
dashdotted,blue, every mark/.append style={solid, fill=white}, mark=diamond\\%
}
\pgfplotscreateplotcyclelist{my accuracy five colors}{%
solid,black, every mark/.append style={solid, fill=white},mark=pentagon\\%
solid,blue, every mark/.append style={solid, fill=white},mark=diamond*\\%
loosely dashed,purple, every mark/.append style={solid, fill=white}, mark=square*\\%
dashdotted,red, every mark/.append style={solid, fill=white}, mark=triangle\\%
solid,green!80!black, very thick, every mark/.append style={solid, fill=white}, mark=star\\%
}
\pgfplotscreateplotcyclelist{my cl six colors}{%
solid,black, every mark/.append style={solid, fill=white},mark=pentagon\\%
loosely dashed,violet, every mark/.append style={solid, fill=white}, mark=square*\\%
solid,blue, every mark/.append style={solid, fill=white},mark=diamond*\\%
dashdotted,red, every mark/.append style={solid, fill=white}, mark=triangle\\%
solid,green!80!black, very thick, every mark/.append style={solid, fill=white}, mark=star\\%
dashdotted,smalt(darkpowderblue), every mark/.append style={solid, fill=ultramarineblue}, mark=*\\%
}
\pgfplotscreateplotcyclelist{my four fsr colors}{%
solid,hotmagenta, every mark/.append style={solid, fill=hotmagenta}, mark=triangle\\%
dashdotted,red, every mark/.append style={solid, fill=white}, mark=star\\%
solid,vviolet, every mark/.append style={solid, fill=white},mark=square*\\%
dashdotted,blue, every mark/.append style={solid, fill=white}, mark=diamond\\%
}
\pgfplotscreateplotcyclelist{my three fsr colors}{%
solid,hotmagenta, every mark/.append style={solid, fill=hotmagenta}, mark=triangle\\%
dashdotted,red, every mark/.append style={solid, fill=white}, mark=star\\%
dashdotted,blue, every mark/.append style={solid, fill=white}, mark=diamond\\%
}
\pgfplotscreateplotcyclelist{my two fsr colors}{%
dashdotted,red, every mark/.append style={solid, fill=white}, mark=star\\%
dashdotted,blue, every mark/.append style={solid, fill=white}, mark=diamond\\%
}

\begin{document}
\title{\LARGE A Crowd-enabled Solution for Privacy-Preserving and Personalized Safe Route Planning for Fixed or Flexible Destinations (Full Version)}
\author{Fariha Tabassum Islam, Tanzima Hashem, and Rifat Shahriyar
\IEEEcompsocitemizethanks{\IEEEcompsocthanksitem F.T. Islam, T. Hashem, and R. Shahriyar are with the Department of Computer Science and Engineering, Bangladesh University of Engineering and Technology, Dhaka, Bangladesh
\IEEEcompsocthanksitem E-mail: fariha.t13@gmail.com, tanzimahashem@cse.buet.ac.bd, rifat@cse.buet.ac.bd}%
}

\markboth{Journal of \LaTeX\ Class Files,~Vol.~14, No.~8, August~2015}%
{Shell \MakeLowercase{\textit{et al.}}: Bare Demo of IEEEtran.cls for Computer Society Journals}

\IEEEtitleabstractindextext{%
\begin{abstract}
\oc{Ensuring travelers' safety on roads has become a research challenge in recent years. We introduce a novel safe route planning problem and develop an efficient solution to ensure the travelers' safety on roads.} Though few research attempts have been made in this regard, all of them assume that people share their sensitive travel experiences with a centralized entity for finding the safest routes, which is not ideal in practice for privacy reasons. Furthermore, existing works formulate safe route planning in ways that do not meet a traveler's need for safe travel on roads. Our approach finds the safest routes within a user-specified distance threshold based on the personalized travel experience of the knowledgeable crowd without involving any centralized computation. We develop a privacy-preserving model to quantify the travel experience of a user into personalized safety scores. \nc{Our algorithms, direct and iterative for finding the safest route further enhance user privacy by minimizing the exposure of personalized safety scores with others.} \oc{Our safe route planner can find the safest routes for individuals and groups by considering both a fixed and a set of flexible destination locations.} 
Extensive experiments using real datasets show that our approach finds the safest route in seconds. \nc{Compared to the direct algorithm, our iterative algorithm requires 47\% less exposure of personalized safety scores.}
\end{abstract}

}

\maketitle

\IEEEdisplaynontitleabstractindextext
\IEEEpeerreviewmaketitle

\section{Introduction} \label{sec:intro}

\oc{Ensuring safe travel on roads is essential for the development of a safe city. \nc{While traveling on roads in any mode (e.g., walking, cycling or driving), people face many inconveniences like theft, robbery, pick-pocketing, and accidents; women face harassment like eve-teasing and unwanted physical touch~\cite{natarajan2016crime,spicer2016street,news0,news1,news2}.} 
The shortest or the fastest route is not always the best choice. People would like to travel a little bit longer on a safer route that avoids those inconveniences. Journey planners like Google or Bing Maps do not show the risky roads to travelers. Since the safety of a road may change with time, it is not easy for a traveler to know the safest route for traveling from a source to a destination location. To meet the traveler's need on roads, we introduce a safe route planner that finds the safest routes (SRs) with crowdsourced data and computation.}

\oc{Our safe route planner supports four important query types: (i) \emph{safest route (SR) query}. (ii) \emph{flexible safest route (FSR) query}, (iii) \emph{group safest route (GSR) query}, and (iv) \emph{group flexible safest route (GFSR) query}. An SR query finds the SR between a source-destination pair within a distance constraint. Sometimes a user may have the flexibility for the destination; for example, a user would be happy to visit any of the branches of a superstore within a distance constraint if the safety level of the route to reach the superstore is increased. Inspired by this scenario, an FSR query finds the SR within a distance constraint by considering a fixed source and a set of destination locations. On the other hand, the GSR and the GFSR queries extend the SR and FSR queries for groups, respectively. A group of people may want to meet for a variety of purposes; sometimes their destination is fixed (e.g., a specific restaurant), and sometimes it is flexible (e.g., a set of restaurants). A GSR query finds the set of SRs from the independent source locations of the group members to the fixed destination, whereas a GFSR query finds the set of SRs from the independent source locations of the group members by considering a set of flexible destinations. In Section~\ref{problem-formulation}, we explain these queries with examples.}   

The data needed for computing the SRs may come from official reports and the personal travel experiences of the crowd. The latter is more valuable than the former one due to its recency and adequacy. However, travel experiences are often sensitive and private data, and people, especially women, do not feel comfortable sharing their detailed travel experiences and harassment data with others\oc{~\cite{DBLP:conf/chi/protibadi14}}. These factors have inspired us to develop a privacy-enhanced safe route planning system by not sharing the personalized travel experiences of the crowd with a centralized entity or others.

Our approach ensures the privacy of crowd data and personalizes the safety score (SS) of a user's travel experience (both safe and unsafe) with respect to the user's travel pattern. If two users face the same unsafe event on two different roads, then these roads may have different SSs considering the frequency and recency of the users' visits on those roads. Ignoring the personal travel pattern of the users would reduce the quality of data and the accuracy of the query answer. We develop a model to quantify a user's travel experience for a visited area into a \emph{personalized safety score} (pSS) based on different parameters like frequency and recency of the user's visits, location, time, and type of inconveniences faced. Users store their pSSs of their known areas on their own devices or any other private storage (e.g., cloud storage) and use them to find the SRs for others. The transformation of a user's travel experience into a pSS is a one-way mapping. From the revealed pSS of a user, it is not possible to pinpoint the type of incident faced by the user. It may only allow an adversary to infer high-level information on a user's travel experience (e.g., a user has encountered an unsafe event without knowing the unsafe event type).

To further enhance user privacy, we minimize the amount of pSS information shared to evaluate the SRs. We develop efficient query processing algorithms that find the SRs from the refined search space and minimize the exposure of pSS information. Since the number of possible routes between a source-destination pair is extremely high, a naive algorithm cannot find the SRs in real time. Our search space refinement techniques allow our query processing algorithms to find the SRs with significantly reduced processing overhead. 

Every user is not familiar with all roads, and it is also not feasible to involve a user for all queries. For a specific SR query or its variant, we identify the users who are familiar with the query relevant area and select them as query-relevant group members. The trustworthiness of the query answer depends on the overall knowledge of the selected query-relevant group members. To show the credibility of the answer, we present a new measure called \emph{confidence level}~\cite{DBLP:journals/imwut/HashemHSM18,DBLP:journals/percom/MahinHK17} in the context of finding the SRs and variants.

Existing safe route planners involve a centralized entity to find the SRs using crime or accident data collected from reports~\cite{DBLP:journals/is2016/Urban-navigation} or crowd~\cite{DBLP:conf/www2014/SocRoutes} or both~\cite{DBLP:conf/gis2011/crowdsafe,DBLP:conf/gis2014/treads,r:safe:CrowdAdaptive}. They have major limitations:
\begin{compactitem}
\item  Ignore the privacy issues of the crowd harassment and incident data and thus suffer from data scarcity problem. Missing incident data can cause a system to return a route that is not actually safe and put a traveler at risk.
\item Do not personalize the crowd’s travel experiences by considering a user’s travel pattern, which is essential to improve the accuracy of the query answer.
\item Do not consider individual distances associated with different SSs for ranking the routes. For example, if two routes have the same lowest SS, then the route for which a user has to travel less distance with the lowest SS is the safest one, though its total distance might be greater than that of the other route.
\item Do not show any measure to represent the trustworthiness of the identified SRs.
\end{compactitem}
 
In recent years, the increase of the computational power and storage in smartphones has enabled researchers to envision for crowdsourced systems~\cite{DBLP:journals/imwut/HashemHSM18,DBLP:journals/percom/MahinHK17}. To the best of our knowledge, we propose the first privacy-enhanced and personalized solution for safe route planning with crowdsourced data and computation. Our solution overcomes the limitations of existing route planners. Our contributions in this paper are as follows:
\begin{compactitem}
 \item We present a model to quantify a user's travel experiences into irreversible pSSs and modify the indexing technique, $R$-tree to store pSSs. Based on pSSs, we design a privacy-enhanced crowd-enabled solution for the SR queries and variants.
 
 \item We select the users who have the required knowledge in a query relevant area, and we guarantee the credibility of the query answer evaluated based on the data of the selected group members in terms of the confidence level.
 
 \item We develop optimal algorithms, direct and iterative, to efficiently evaluate the SRs. The direct algorithm reveals group members' pSSs only for the query relevant area. The iterative one further reduces the amount of shared pSSs at the cost of multiple communications per group member. 
 
 \item \oc{We generalize our direct and iterative algorithms to efficiently process the SR query and its variants: FSR, GSR, and GFSR queries. We show that the direct application of the direct and iterative algorithms to find SRs between a source-destination pair to evaluate SR query variants incur excessive processing overhead in most cases.} 
 
 \item We run extensive experiments with real datasets and evaluate the effectiveness and efficiency of our approach. 
\end{compactitem}

\oc{This paper extends the work in~\cite{DBLP:conf/icde/IslamHS21}, where we introduced a novel SR query and proposed the first privacy-enhanced and personalized solution to solve those queries with crowdsourced data and computation. In this paper, we enhance the work in the following ways: (i) we improve our safe route planner by introducing SR query variants: FSR, GSR, and GFSR queries, (ii) we provide generalized direct and iterative algorithms for efficient processing of the SR query and its variants, \nc{(iii) we show the complexity analysis and performance analysis of our modified $R$-tree, (iv) we provide the formal privacy attacker model and privacy proof,} (v) we present new experimental analysis to show the efficiency of our generalized direct and iterative algorithms for the SR query variants, and \nc{the effectiveness of finding SRs over the shortest routes and the safest routes without any distance constraint}. 
}

\section{Problem Formulation}\label{problem-formulation}
The road network $N=(V,E)$ consists of a set of vertices $V$ and a set of road segments $E$. The vertices represent the start or the end or the intersection points of roads. An edge $e_{ij} \in E$ \oc{representing a road segment} connects the vertex $v_i$ to the vertex $v_j$, where $v_i,v_j\in V$. 
A route $R$ consists of a sequence of vertices $R = (v_{i_1},v_{i_2},\ldots,v_{i_{|R|}})$, where $e_{ {i_{k-1}} {i_k} }\in E$. The total distance $dist(R)$ of $R$ is the summation of distances of all edges in $R$. 

The total space is divided into grid cells. The knowledge score (KS), the pSS, and the SS are computed for each grid cell area and are defined as follows:
    \begin{definition}
    \emph{A knowledge score (KS):} The KS of a user for a grid cell area represents whether the user has visited the area of a grid cell. This KS is 0 if the user has not visited the area in the last $w$ days, 1 otherwise, \oc{where $w$ is an integer greater than 0}.     
    \end{definition}    
    \begin{definition}
    \emph{A personalized safety score (pSS): Given the safety score bound $[-S, S]$, the pSS of a grid cell area represents a user's travel experience in the area and is quantified between $-S \leq pSS \leq S$.}
    \end{definition}    
    \begin{definition}
    \emph{A safety score (SS):} Given \label{def_ss}a set of pSSs $\Psi_1,\Psi_2,\ldots,\Psi_n$ of $n$ users for a grid cell area, the SS of the grid cell area is computed as $\floor*{ \frac{ \Psi_1+\Psi_2+\ldots+\Psi_n}{n}}$.    
    \end{definition}    

To make the SS measure independent of the number of users who know about an area, we take the average of the pSSs instead of adding them together. The number of users whose pSSs are used to find the SS is considered to determine the credibility of the safest route (Section~\ref{confidence-level}). \oc{An edge representing a road segment may have multiple SSs if it passes through multiple grid cells with different SSs.}

\emph{SS-based route ranking.} The SS of route $R$ is the minimum of all SSs associated with the edges of $R$. The intuition behind considering the minimum SS instead of the average SS of the route is that even a small distance of road with a bad SS may put a traveler at risk. The route that has the largest minimum SS among all possible routes between a source-destination pair is considered as the SR. \nc{If the minimum SS of two routes is the same, then we consider the smallest SS for which associated distances of two routes differ.} The route that has the smallest associated distance for the considered SS is the SR. We formally define the SR within distance constraint $\delta$ as follows: 

\begin{definition}
\emph{A safest route (SR):}
\oc{Given a road network $N(V, E)$, distances and SSs of road segments, a source location $s$, a destination location $d$ and a distance constraint $\delta$, the safest route $SR$ between $s$ and $d$ is a route such that $dist(SR) \leq \delta$ and $SR$ is at least as safe as $R$, where $R$ is any other route between $s$ and $d$ having $dist(R) \leq \delta$.}
\end{definition}

\oc{\emph{SS based route-set ranking.} The SS of a set of routes $\mathcal{H}=\{R_1, R_2, ..., R_n\}$ is the minimum of the SSs of all routes in $\mathcal{H}$. Assume, $\mathcal{H}$ and $\mathcal{H}^\prime$ are two set of routes, and $R$ and $R^\prime$ are the least safe routes of $\mathcal{H}$ and  $\mathcal{H}'$, respectively. If $R$ is safer than $R^\prime$, then  $\mathcal{H}$ is safer than $\mathcal{H}^\prime$; otherwise, $\mathcal{H}^\prime$ is safer. }

\begin{figure*}[htb!]
\begin{center}
    \centering
    \includegraphics[height=0.7cm]{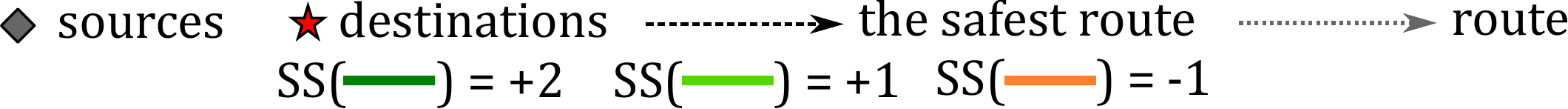}\\
    \subfloat[An SR query\label{fig:SR}]{\includegraphics[width=0.48\textwidth]{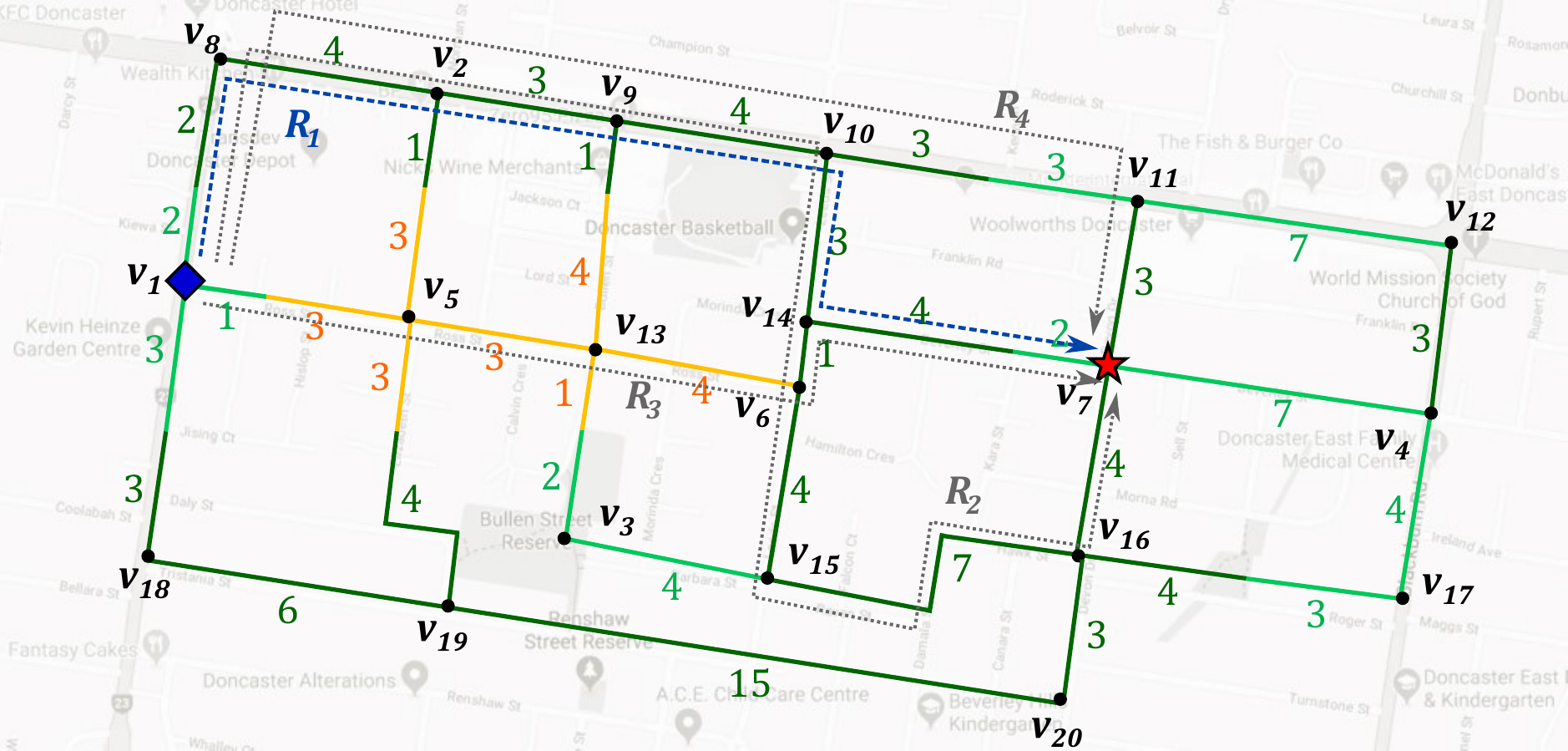}}
    \hfill
    \subfloat[An FSR query\label{fig:FSR}]{\includegraphics[width=0.48\textwidth]{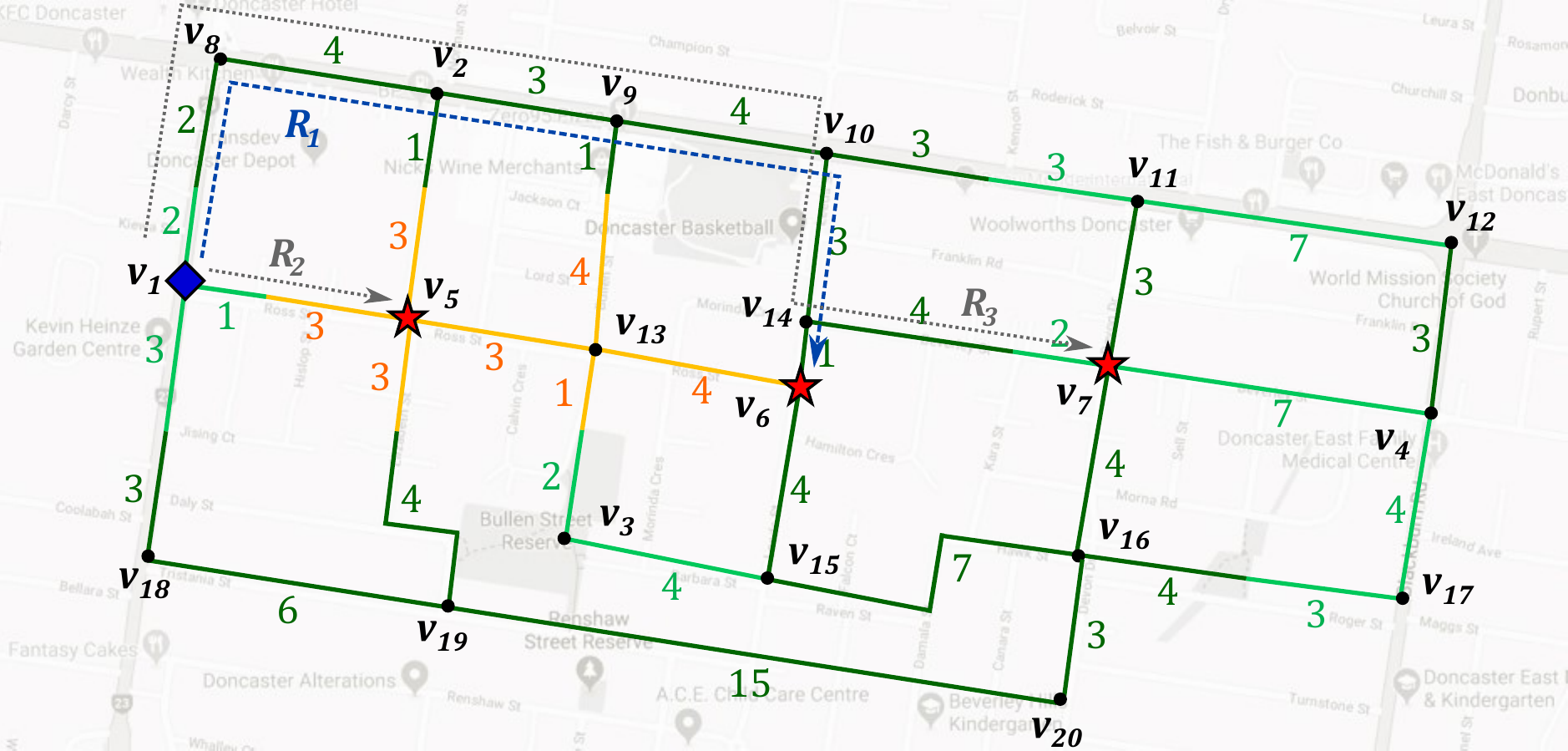}}
\newline
\noindent
    \subfloat[A GSR query\label{fig:GSR}]{\includegraphics[width=0.48\textwidth]{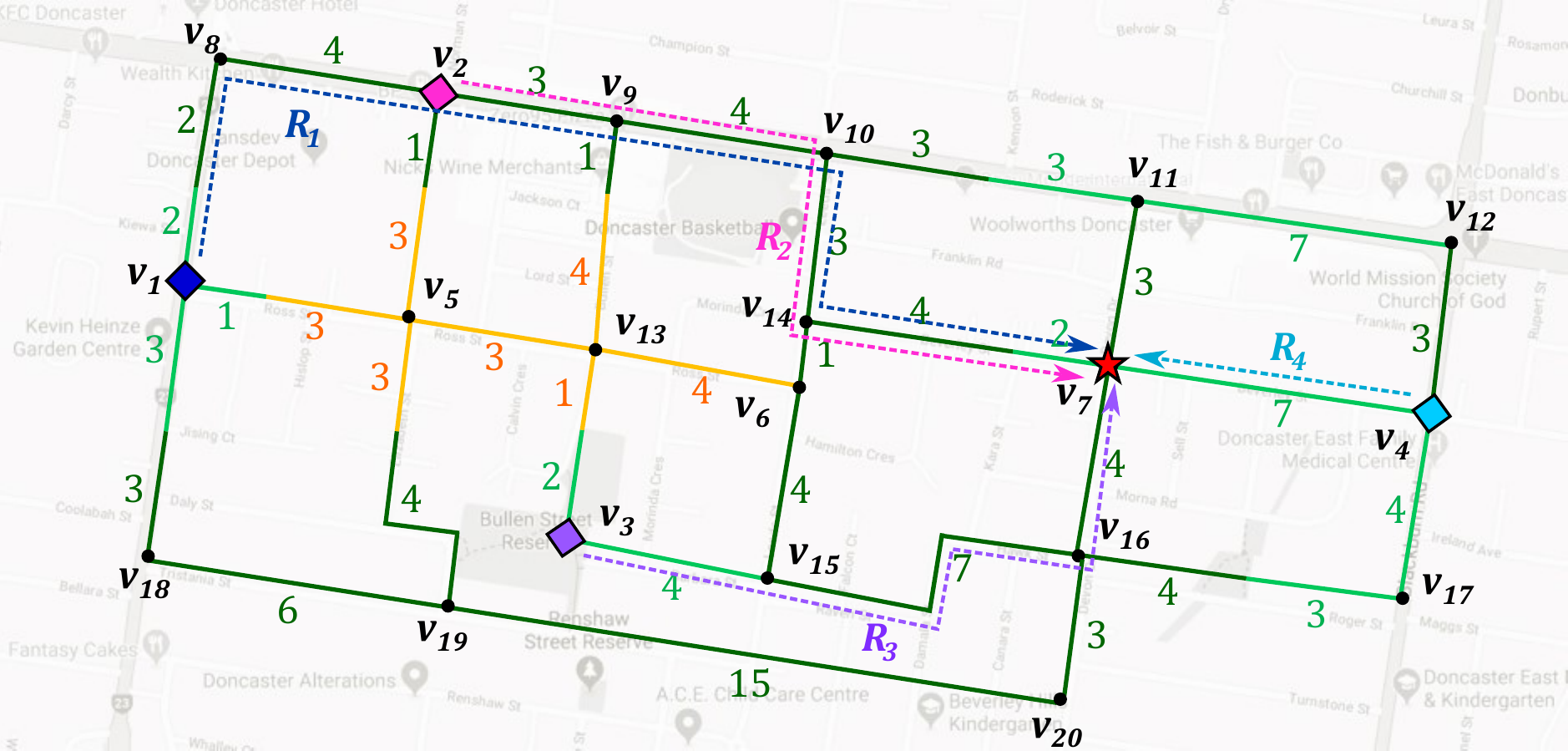}} 
    \hfill
    \subfloat[An GFSR query\label{fig:GFSR}]{\includegraphics[width=0.48\textwidth]{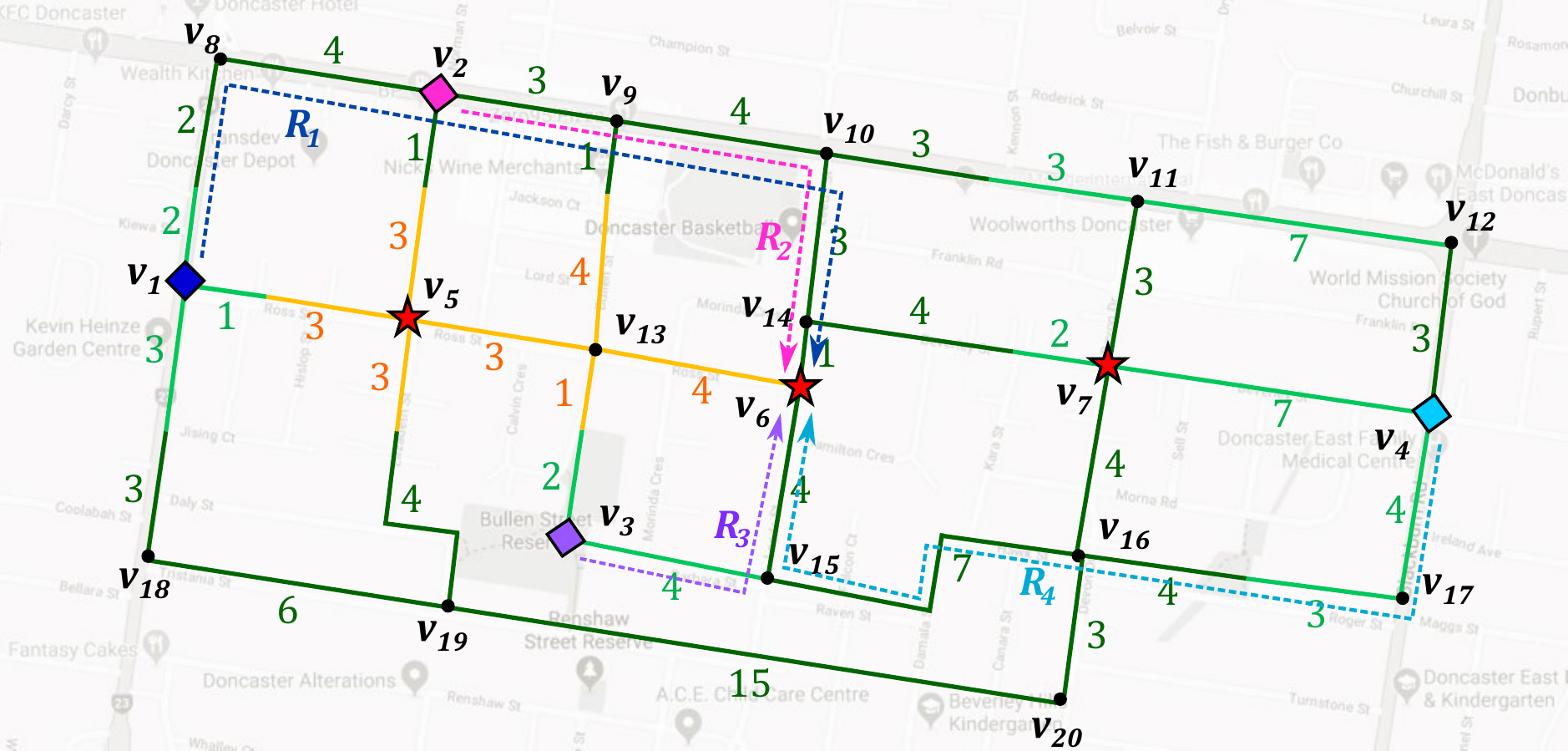}}
\end{center}
\caption[caption]{Examples of various queries in our safe route planner for a small road network where $\delta=25km$. The safety score (SS) of a road segment is represented by its color. The numbers in the road network represent the length of a road segment.}
\label{fig:all-queries}
\end{figure*}

\oc{Now we formulate our safe route planner that addresses the SR query and its variants: the FSR query, the GSR query and the GFSR query}.
\begin{definition}\label{def:gfsr}
\emph{A Safe Route Planner:} 
\oc{Given a road network $N(V, E)$, distances and SSs of road segments, a set of $n$ source locations $L_S = \{s_1, s_2, \dots, s_n\}$, a set of $m$ destination locations $L_D = \{d_1, d_2, \dots, d_m\}$ for $n,m\geq 1$ and a distance constraint $\delta$, the safe route planner returns a set of routes $\mathcal{H^*} = \{SR_1, SR_2, \dots, SR_n\}$ such that the following conditions are satisfied:}   
\begin{enumerate}
\item \oc{For $m \geq 1$, $SR_i$ is the SR between $s_i$ and $d$, where $d \in L_D$  is the same for all returned routes $SR_1$, $SR_2$, $\dots$, $SR_n$, and} 
\item \oc{For $m > 1$, $\mathcal{H^*}$ is safer than $\mathcal{H}^\prime =\{ {SR^\prime}_1$, ${SR^\prime}_2$, $\dots$, ${SR^\prime}_n\}$ where ${SR^\prime}_i$ is the SR between $s_i$ and another destination $d^\prime \in L_D-\{d\}$}
\end{enumerate}
\end{definition}
\oc{When $n=1$ and $m=1$ in Definition~\ref{def:gfsr}, then the safe route planner evaluates a \emph{safest route (SR) query}. When $n=1$ and $m>1$, then it evaluates the \emph{flexible safest route (FSR) query}. On the other hand, when $n>1$, the safe route planner evaluates the \emph{group safest route (GSR) query} for $m=1$, a \emph{group flexible safest route (GFSR) query} for $m>1$. Fig.~\ref{fig:all-queries} shows examples of these four query types.}

\oc{Fig.~\subref*{fig:SR} shows an example of an SR query, where the source is $v_1$ and the destination is $v_7$. The SR query returns $R_1$ as the SR within distance constraint. Fig.~\subref*{fig:SR} also shows $R_2$, $R_3$, and $R_4$ that are some of the possible routes from $v_1$ to $v_7$. The length of $R_2$ is 34 units, which does not satisfy $\delta$, and thus cannot be the SR. $R_1$ is safer than $R_3$ as the smallest SSs associated with them are +1 and -1, respectively. In $R_3$, a user has to pass through a riskier area compared to $R_1$. Though both $R_1$ and $R_4$ have the same length (i.e., 24 units) and smallest SS (+1), $R_1$ is the SR because $R_1$ has the smaller distance (i.e., 4 units) associated with +1.}

\oc{Fig.~\subref*{fig:FSR} shows an example of an FSR query, where $v_1$ is the source, $\{v_5,v_6,v_7\}$ is a set of destinations. The set of flexible destinations may represent different branches of a superstore or ATM booths of a bank, and the user is happy to travel to any of these destinations within the distance constraint that maximizes the route safety. The FSR query returns destination $v_6$ and $R_1$ as the SR to $v_6$. The reason is as follows. In this example, $R_1$, $R_2$, and $R_3$ are the SRs (within $\delta$) from source $v_1$ to destinations $v_6$, $v_5$, and $v_7$, respectively. $R_1$ is safer than $R_2$ as the minimum SS of $R_1$ and $R_2$ are +1 and -1, respectively. The minimum SS is the same (+1) for both routes $R_1$ and $R_3$. $R_1$ is still safer than $R_3$ because the length associated with +1 is 2 and 4 for $R_1$ and $R_3$, respectively.}

\oc{Fig.~\subref*{fig:GSR} shows an example, where a group of users from different locations are planning to meet at a fixed destination (e.g., friends meeting at their favorite restaurant) and requests a GSR query. Here $v_1$, $v_2$, $v_3$, and $v_4$ are source locations of each user of a group and $v_7$ is their destination. The GSR query returns the SRs $R_1$, $R_2$, $R_3$, and $R_4$ from $v_1$, $v_2$, $v_3$ and $v_4$, to destination $v_7$ respectively.}
    
\oc{Fig.~\subref*{fig:GFSR} shows an example of a GFSR query. In this scenario, a group of users might want to meet at any of the set of specified destinations (e.g., restaurants) that can be reached via safer routes compared to others within a distance limit. Here $v_1$, $v_2$, $v_3$ and $v_4$ are the source locations of the users of a group and $\{v_4,v_5,v_6\}$ is the set of their preferred destinations. The GFSR query returns $v_6$ as the safest destination and $R_1$, $R_2$, $R_3$ and $R_4$ as the SRs from $v_1$, $v_2$, $v_3$ and $v_4$, respectively. Here, $v_5$ is not returned as the answer because the users have to go through unsafe road segments of SS -1 to reach it, which can be avoided in case of $v_6$. Moreover, $v_7$ is also not selected because some users have to pass through comparatively unsafer road segments of SS +1 ($v_{14}$ to $v_7$) which can be avoided in $v_6$.}

\nc{\emph{Privacy-enhanced safe route planner.} Our \emph{privacy-enhanced} safe route planner aims to hide the unsafe event types that a user has faced from an adversary, i.e., the centralized server and other users. A pSS only reveals high-level information, like a user encountered an unsafe event, but not the type. Our solution also aims to minimize the number of revealed pSSs to enhance a user's privacy. We detail our privacy model in Section~\ref{privacy-analysis}.}

\section{Related Works} \label{related-works}

\begin{small}
\begin{table*}[t]
\caption{A comparative analysis with existing safe route planners}
\label{tab:comp_related}
\vspace{-1mm}
\centering
\begin{tabular}{|=l||+C{1.8cm}|+C{1.8cm}|+C{0.7cm}|+C{6.9cm}|+C{0.4cm}|+C{0.7cm}|+C{0.8cm}|}

 \hline
 \multirow{2}{*}{} & \multicolumn{5}{c|}{Problem Settings} & \multirow{2}{*}{\barr Pri-\\vacy\earr} & \multirow{2}{*}{\barr Effi-\\ciency\earr} \\ \cline{2-6}
    & \barr Source- \\ Destination\earr & \barr Safety \\ Level\earr & pSS & Objective & $\delta$ &  &  \\ \hline \hline

 \cite{DBLP:journals/is2016/Urban-navigation} & Single & Multiple & $\times$ & Provide multiple routes with trade off between SS and total distance & $\times$ & $\times$ & \checkmark \\ \hline

 \cite{DBLP:conf/www2014/SocRoutes} & Single &  Safe/ Unsafe & $\times$ & Minimize the travel in unsafe regions & $\times$ & $\times$ &  $\times$\\ \hline

 \cite{DBLP:conf/gis2011/crowdsafe,DBLP:conf/gis2014/treads} & Single &  Multiple & $\times$ & Minimize the weighted combination of SS and total distance & $\times$ & $\times$ &  $\times$\\ \hline

 \cite{DBLP:conf/icde2015/safest-path-via-safe-zones} & Single & Safe/ Unsafe & $\times$ & Minimize the travel in unsafe regions & $\times$ & $\times$ &  \checkmark\\ \hline

 \cite{sarraf2018data,sarraf2018data2} & Single & Multiple & $\times$ & Minimize the weighted combination of accident count and travel time & $\times$ & $\times$ &  $\times$ \\ \hline

 \cite{sarraf2020integration} & Single & Multiple & $\times$ & 
 Select a route from a set of alternative routes based on users' preferences of safety, distance and travel time & $\times$ & $\times$ &  N/A \\ \hline


 Ours & Single/ Multiple &  Multiple & \checkmark & \oc{{Maximize the minimum SS of the route and then minimize the individual distances associated with the SSs in the increasing order of SSs}} & \checkmark & \checkmark  &  \checkmark\\ \hline
\end{tabular}%
\vspace{-3mm}
\end{table*}
\end{small}

\subsection{Safe Route Planners}
\label{related:safe}

Though researchers attempted to solve the safe route planning problem, the works have major limitations. Table~\ref{tab:comp_related} shows the problem settings and other features of existing works.

\emph{Problem setting.}
None of the existing work considers minimizing the individual distance associated with risky roads. Thus, the problem settings of existing works are not suitable for safe travel on roads. Furthermore, instead of considering the total distance constraint, selecting appropriate weights in~\cite{DBLP:conf/gis2011/crowdsafe,DBLP:conf/gis2014/treads} is not easy since it is not intuitive to determine which weights would meet a user's preferred trade-off between safety and distance for a specific source-destination pair. Again, there is no guarantee that the returned routes in~\cite{DBLP:journals/is2016/Urban-navigation},\nc{\cite{sarraf2018data,sarraf2018data2,sarraf2020integration}} satisfy a user's required preference for safety and distance. \oc{None of the existing works addresses the problems of finding the SRs for flexible destinations or for a group located at different source locations.}

\emph{Privacy.} 
\nc{The unsafe event data for safe route planners may come from crime and accident reports}~\cite{DBLP:journals/is2016/Urban-navigation,r:safe:CrowdAdaptive,r:safe:fns},\nc{\cite{sarraf2018data,sarraf2018data2,sarraf2020integration}} or directly from crowds~\cite{DBLP:conf/gis2011/crowdsafe,DBLP:conf/gis2014/treads,r:safe:CrowdAdaptive,DBLP:conf/www2014/SocRoutes}. \nc{Those reports are not regularly updated, and incomplete because many crimes and accidents go unreported}. Though the crowd knows more and recent information compared to the crime and accident reports, they would not share their incident and harassment data with a centralized service provider if the privacy of their data is not ensured. Thus, one major limitation of existing works is that they suffer from data scarcity issues for privacy reasons and do not have enough data to provide accurate answers.

\emph{Efficiency.} 
None of the existing safe route planning systems except \cite{DBLP:conf/icde2015/safest-path-via-safe-zones,DBLP:journals/is2016/Urban-navigation} 
developed efficient algorithms for large road networks. However, as already mentioned, the problem settings of \cite{DBLP:conf/icde2015/safest-path-via-safe-zones,DBLP:journals/is2016/Urban-navigation} cannot meet a traveler's requirement on roads.

\emph{Other route planners.} 
Variants of orienteering and scheduling problems~\cite{AllahverdiNCK08,VansteenwegenSO11} have been studied for route planning. An orienteering problem finds a route between a source-destination pair that maximizes the total score within a budget constraint, where a score is obtained when the route goes through a vertex. The scheduling problems focus on incorporating temporal constraints in route planning (e.g., visiting locations to perform services in a timely manner). The problem settings of orienteering and scheduling problems are different from an SR query. Furthermore, their solutions do not consider search space refinement~\cite{JahanHSB19} and are not scalable for large road networks. For example, the exact solution of an orienteering problem can be found for a graph of up to 500 vertices~\cite{VansteenwegenSO11}, whereas the real road networks that we use in our experiments have, on average, 24 thousand vertices.

\subsection{\nc{Nearest and Group Nearest Neighbor Queries on Road Networks}}
Extensive studies~\cite{DBLP:conf/vldb/PapadiasZMT03,DBLP:conf/gis/SankaranarayananAS05,DBLP:conf/sigmod/SametSA08,DBLP:conf/icde/0002CGSMG18,DBLP:conf/icde/PapadiasSTM04,DBLP:journals/pvldb/YanZN11,DBLP:conf/aips/AbeywickramaCS20} have been done to efficiently compute the nearest and group nearest neighbors. The nearest neighbor (NN) query finds the nearest POI that has the smallest distance from a user's location, whereas a group nearest neighbor (GNN) query minimizes the aggregate distance of the POI from the locations of the group members. None of the existing solutions for NN or GNN queries have considered ensuring user safety on  roads.  

\subsection{Crowdsourcing} \label{related:crowd}
Crowdsourcing has been widely used for route recommendation\cite{r:crowd:shorest-path,r:crowd-route:crowdplanner} and POI search~\cite{DBLP:journals/imwut/HashemHSM18,DBLP:journals/percom/MahinHK17, DBLP:conf/huc/HashemAKTQ13}, package delivery~\cite{r:crowd:delivery} and indoor mapping~\cite{r:crowd:indoor}. In~\cite{DBLP:journals/imwut/HashemHSM18}, the authors considered protecting the privacy of a user's POI knowledge by minimizing the shared POI information with others. Compared to static POI data, unsafe events' data are more complex and challenging to hide from others. We develop a quantification model to hide the type of incident data using pSS and search space refinement techniques to minimize the shared pSS information.     
\section{System Overview} \label{system-overview}

\begin{figure*}[tb!]
  \centering
  \includegraphics[width=\textwidth]{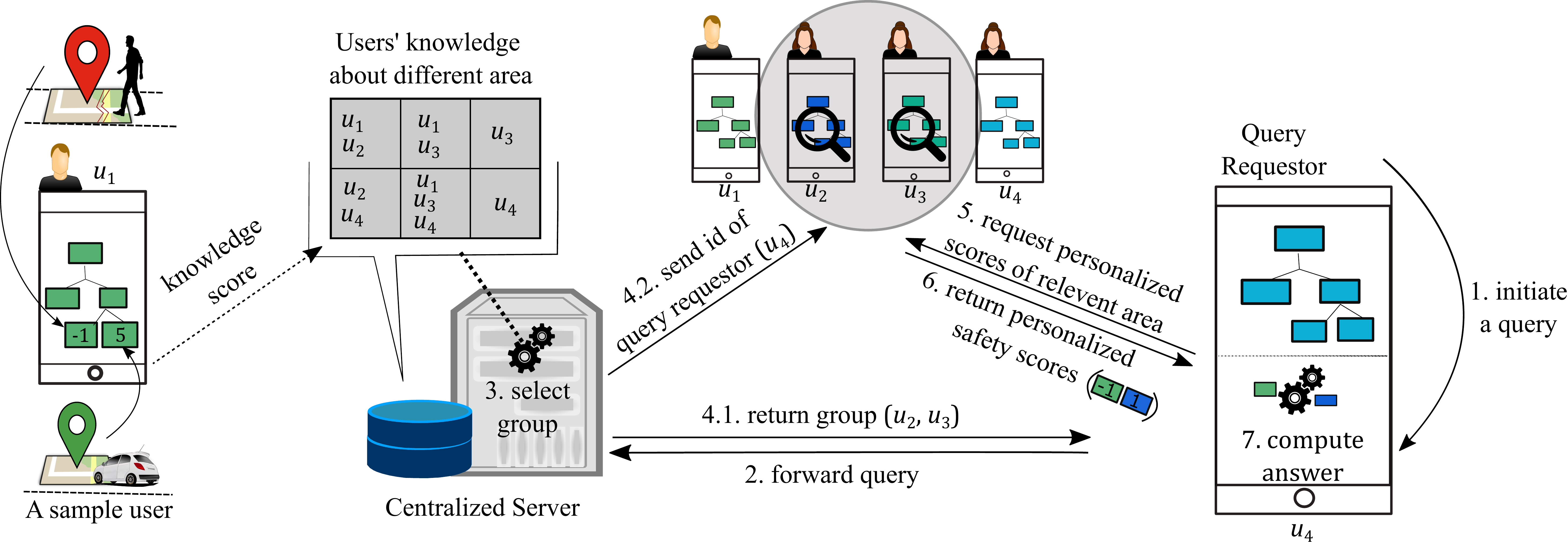}
  \caption{System architecture}
  \label{fig:system-overview}
\end{figure*}

We develop a privacy-enhanced, personalized, and trustworthy solution for safe route planning with crowdsourced data and computation. Fig.~\ref{fig:system-overview} shows the architecture of our system. Users in our system store their pSSs of their visited areas on their own devices. In the case of storage constraints, users can also consider alternative private storage (cloud storage). The users share their KSs with the centralized server. A KS only provides the information that a user has visited the area. A user can also hide the information of her visit to a sensitive area by not setting the corresponding KS to 1 as the user has the control to decide on what the user shares with the centralized server. 

It is not realistic to use the computation power of all users for all queries and ask them whether they know any query relevant area. The availability of KSs allows the centralized server to address this issue. When the centralized server receives a query from a query requestor, it selects a query-relevant group based on the query parameters and the stored KSs of the users. Then the centralized server returns the IDs of the query-relevant group members to the query requestor and sends the identity of the query requestor to the query-relevant group members. 
The query requestor evaluates the query in cooperation with the query-relevant group members without involving the centralized server. The query requestor retrieves pSSs of the query-relevant area from the query-relevant group members, computes the SS of each road using the pSSs of the query-relevant group members, and finds the SR or the SRs depending on the query. \oc{For a GSR or a GFSR query, the source locations of all users of the group are sent to one of the group users who acts as a query requestor; the query is evaluated into the query requestor's device and after that, the relevant part of the query answer is sent to every member of the group. For example, for a GSR query, the SR from a group member's source location to a  fixed destination is sent to that group member's device.}
\section{Quantification of Safety} \label{ss:quantification-of-safety}
\subsection{Limitations of Existing Models.} Existing researches on safe routes have modeled safety in a variety of ways. 
The authors of~\cite{DBLP:conf/gis2011/crowdsafe,DBLP:conf/gis2014/treads} quantify the safety of a road network edge by simply considering the number of crimes in the particular distance buffer area of that edge. They do not consider the recency and the severity of crimes, the ratio between the unsafe visits and 
the safe visits by an individual user, and the fact that the impact of a crime decays with distance. Thus, the quantified SSs of roads in~\cite{DBLP:conf/gis2011/crowdsafe,DBLP:conf/gis2014/treads} fail to model the real-scenarios. The work in~\cite{DBLP:journals/is2016/Urban-navigation} improves the way to find the SS of a road network edge by considering the crime events of the last few days and weighting the crime events based on their distances from the road. None of the above works~\cite{DBLP:conf/gis2011/crowdsafe,DBLP:conf/gis2014/treads,DBLP:journals/is2016/Urban-navigation} allow the SS to vary in different parts of a road network edge, which is possible for long roads. 

In~\cite{r:safe:CrowdAdaptive}, the authors provide a more elaborate model of safety. However, the model suffers from the following limitations: (i) stores historical data and cannot address the constraint of the limited storage of the personal devices, (ii) does not differentiate the weights of crime events based on the frequency of the user's visits, (iii) only considers that the effect of a crime spreads to its nearby places only if no crime occurs there, (iv) does not provide a smooth decay of the effect of older events, rather takes the moving average of the events of the last few days, and discards the impact of previous events, (v) does not consider the severity of a crime event, and (vi) does not allow the SS to vary in different parts of an edge.

\subsection{Our Model.} We develop a model that overcomes the limitations of existing models. In our model, the travel experiences of users are converted into pSSs and then aggregated to infer the SSs of different areas. When a user visits an area, an event occurs. If the user faces an unsafe event, then that event is unsafe; otherwise, it is safe. Our model has the following properties:

\begin{enumerate}
  \item The safety of an area depends on the frequency of the users’ visits.
    If a user visits an area twice and faces unsafe events both times, then intuitively, that area is riskier than another area where a user visits 10 times and faces unsafe events two times among those visits.
    If a user visits an area 5 times safely, then that area is safer than another area that is visited once safely. 
  
  \item The safety of an area also depends on the safety of its nearby places. Therefore, if a user visits an area, the impact of the event is distributed to nearby areas.

  \item The safety of an area depends on the recency of the safe and unsafe events. If a user faces an unsafe event in an area, then that event’s effect decays with time. Similarly, if a user safely visits an area, after some time, that visit’s impact decays, and that is not perceived as safe as before.
  
  \item The safety of an area depends on the type and severity of an unsafe event. 
  
  \item The pSSs are not allowed to grow indefinitely. They are bounded within a maximum and a minimum value so that while aggregating, a single user's experience does not dominate the SS of an area.

  \item A road network edge may go through multiple grid cells and thus, can have different SSs.
\end{enumerate}

An important advantage of our model is that it is storage efficient as it does not store historical visit data of a user. 

\emph{Computation.} Let the impact of a safe event in the occurring area be ${\xi}^+$ and the impact of an unsafe one be ${\xi}^-$, where ${\xi}^+, {\xi}^- \in \mathbb{Z}$. ${\xi}^+$ is the same for all safe events. ${\xi}^-$ varies with the type and the intensity of the unsafe event or inconvenience faced. 

The impact ${\xi} (={\xi}^+/{\xi}^-)$ of an event reduces exponentially in nearby areas and becomes ${\xi}'$ as per the following equation:
${\xi}' = {\xi} * e^{-\frac{dist^2}{2h^2}}$, where the constant $h$ controls the spread of the event. $dist$ represents the distance of the event location from the grid cell. This equation is inspired by the Gaussian kernel density estimation~\cite{DBLP:journals/is2016/Urban-navigation}. 

The pSS, $\Psi$, of an area is bounded within $[-S, S]$ and $\Psi \in \mathbb{Z}$ and $0<{\xi}^+<S$ and $-S<{\xi}^-<0$. If an event occurs in a place for the first time then $\Psi = {\xi}$. If another event occurs there, then $\Psi = \Psi+{\xi}$. If an event ${\xi}$ occurs nearby, whose effect is ${\xi}'$ here, then $\Psi = \Psi+{\xi}'$. If $\Psi> S$ then $\Psi = S$ and if $\Psi < -S$ then $\Psi = -S$. Initially, $\Psi$ is set to \emph{unknown}.

A pSS decays in every $\Delta_d$ days. If the decay rate is $r_d$ and $\Psi \neq 0$, then after every $\Delta_d$ days, $\Psi$ becomes $\Psi = \Psi * r_d$, where $0<r_d<1$ and $r_d \in \mathbb{R}$. Therefore, the decay of older events' impacts is smooth. For example, if $r_d=0.8$ and $\Delta_d = 2$, then $\Psi=3$ becomes 2.4  after two days, and becomes 1.92 after two more days.

The values of parameters ${\xi}^+$, ${\xi}^-$, $S$, $\Delta_d$ and $r_d$ are the same for all users and decided centrally. For each grid cell, our model stores only two values: the pSS and when that pSS was last updated. Therefore, this model is storage-efficient and suitable for smart devices. The SS of an area is computed from the shared pSSs of the users (Definition~\ref{def_ss}).  

\nc{Our model allows incorporating the severity of unsafe event types in terms of impacts. For example, the impact of pick-pocketing and robbery are not the same. Similarly, the impact of accidents may vary depending on the underlying cause, like road sinuosity or slipperiness. The safest route may also not remain the same in different contexts (e.g., time and weather of the day, lightning condition, travel mode). Thus, users can store pSSs for different contexts (e.g., time and weather of the day, lightning condition, travel mode) and share them to compute context-specific safest routes.} 
\section{Indexing User Knowledge} \label{indexing-user-knowledge}


A user stores the pSS for every visited grid cell in the local storage and accesses it for evaluating the SR query. The centralized server stores the KSs of users for every grid cell and uses them for computing query-relevant groups. For efficient retrieval of pSSs and KSs, we use indexing techniques: local and centralized, respectively.

\subsection{Local Indexing.}
Storing pSSs for the whole grid in a matrix would be storage-inefficient because a user normally knows about some parts of the grid area. We adopt a popular indexing technique $R$-tree~\cite{rtree} for storing pSSs of the visited grid cells. The underlying idea of an $R$-tree is to group nearby spatial objects into minimum bounding rectangles (MBRs) in a hierarchical manner until an MBR covers the total space. 

For every visited grid cell, a user stores its pSS and the time of its last update. The last update time is required for decaying the pSS. To reduce the storage overhead, we combine nearby adjacent grid cells with an MBR, where the grid cells have the same SS and the difference between the last update time of two cells does not exceed a small threshold. We call this MBR as a supercell and each leaf node of an $R$-tree represents a supercell. Each leaf node stores the information of the coordinates of MBR, the pSS, and the average of the last update time of the considered grid cells of a supercell. The supercells are recursively combined into MBRs. The intermediary nodes of the $R$-tree store the coordinates of the MBR. The MBR of the root node of the $R$-tree represents the total grid area. Fig.~\ref{fig:local-indexing} shows an example of a grid and the  corresponding $R$-tree. For the sake of clarity, we do not show the last update times in the figure.

\begin{figure}
\centering
    \subfloat[The pSSs for a 4$\times$4 grid is stored in a modified $R$-tree\label{fig:local-grid}]{
        \scalebox{0.28}
        {\includegraphics[]{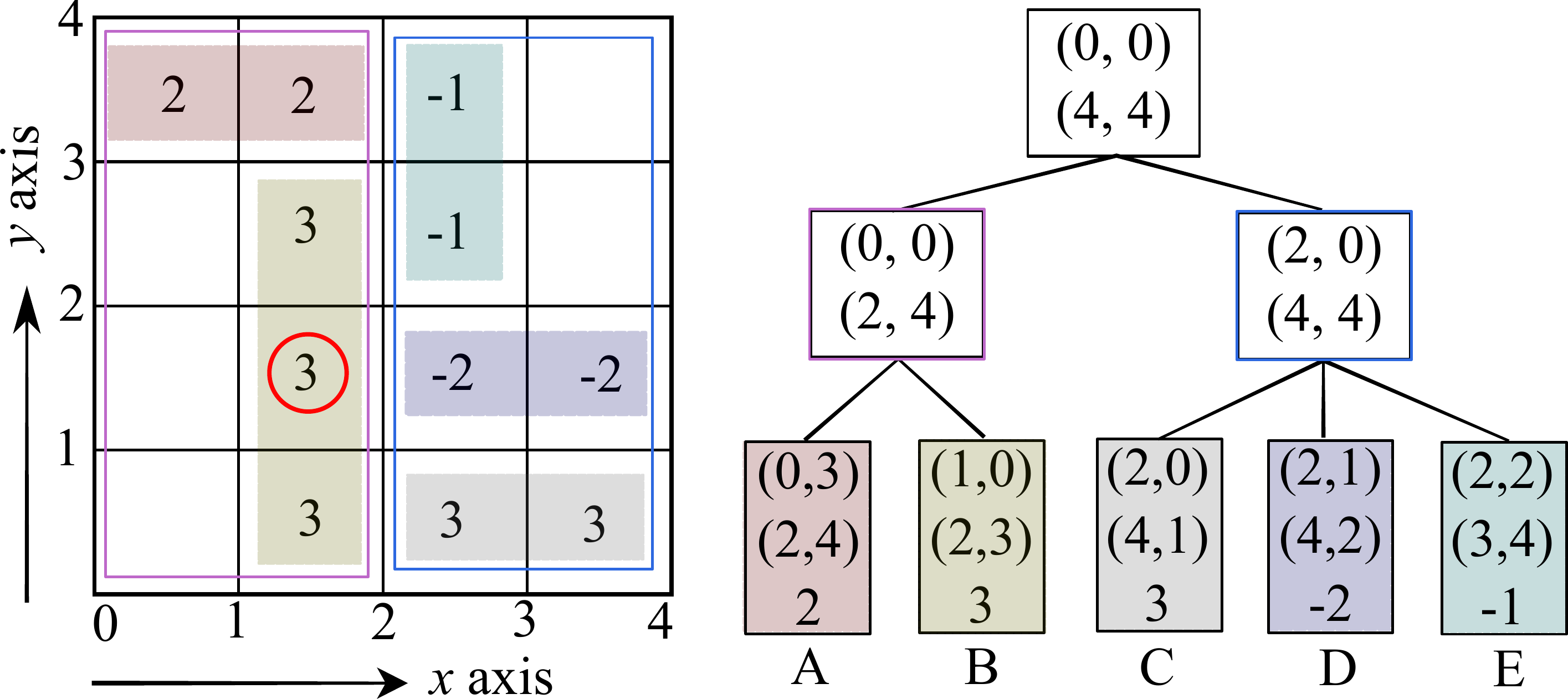}}
    }
    \vspace{2mm}
    \subfloat[A pSS changed from 3 to -2 and is updated in the  $R$-tree\label{fig:local-grid-update}]{
        \scalebox{0.28}
        {\includegraphics[]{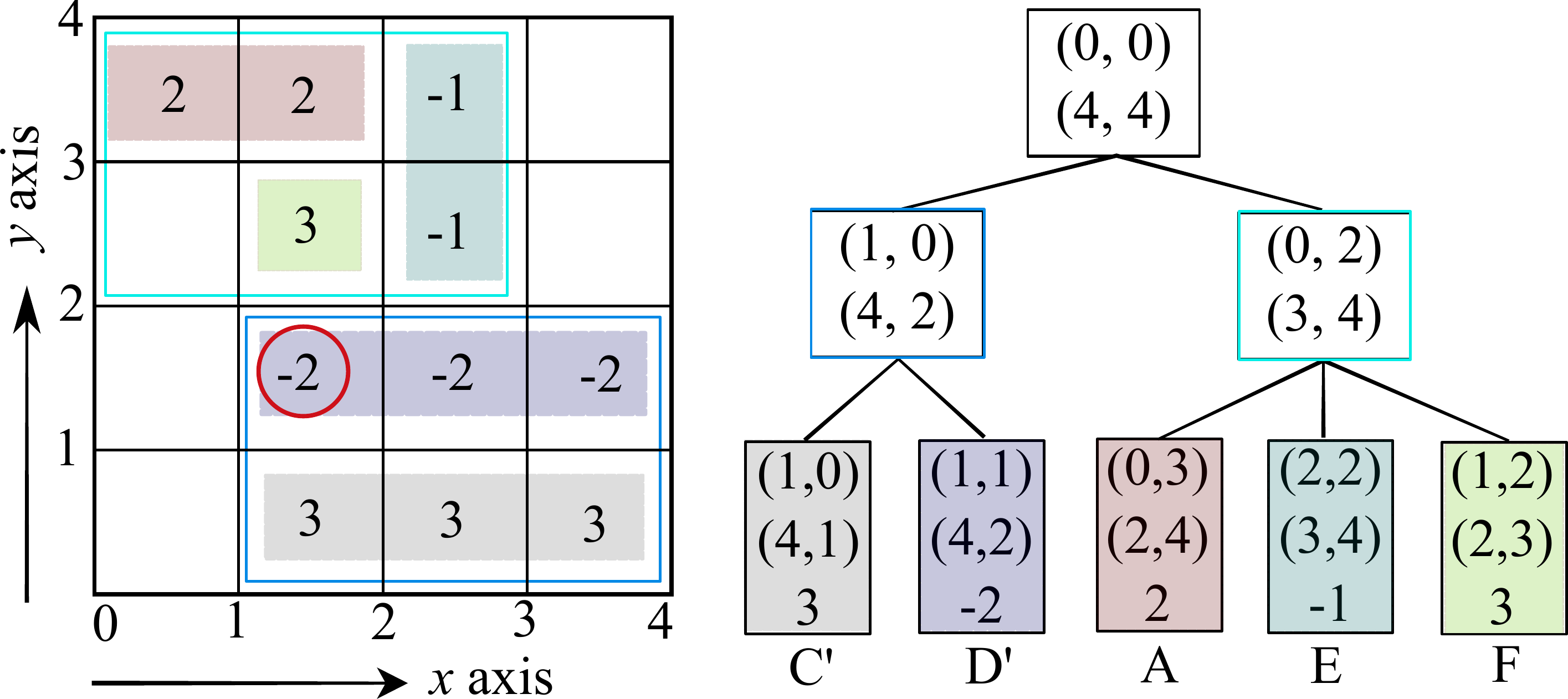}}
    } 
    \caption{A user's pSSs is stored in a modified $R$-tree}
	\label{fig:local-indexing}
\end{figure}

\emph{Supercell generation.} A traditional $R$-tree only considers the location of the spatial objects for grouping, whereas we consider the location, the pSS, and the last update time of the grid cells for grouping them into supercells. To compute the non-overlapping supercells, we scan the grid cells twice: row-wise and column-wise. For row-wise (or column-wise) scan, we maximize the number of grid cells included in a supercell row-wise (column-wise) and then take the supercells of the scan (row-wise or column-wise) that generates the minimum number of supercells. After computing the supercells for the leaf nodes, we insert them into a traditional $R$-tree. 
\nc{
For example, Table~\ref{tab:rowscan-colscan} shows the supercells that are created from row-wise and column-wise scans in Fig.~\ref{fig:local-grid}. 
Since both scans generates five supercells, we choose the supercells generated from the row-wise scan.
}
\begin{table}[!hbt]
\caption{\nc{Supercells generated from row-wise and column-wise scans in the 4$\times$4 grid of Fig.~\ref{fig:local-grid}}}
\label{tab:rowscan-colscan}
\centering
\nc{
\begin{tabular}{|C{1.2cm}|C{0.24\columnwidth}||C{1.2cm}|C{0.24\columnwidth}|}
    \hline
    \multicolumn{2}{|c||}{Row-wise scan} & \multicolumn{2}{c|}{Column-wise scan} \\ \hline
    Row No. & Supercells & Col No. & Supercells \\ \hline
    1 & A [(0,3) (2,4) +2] 
        & 1 & [(1,0) (4,1) +3] \\ \hline
    2 & B [(1,0) (2,3) +3] 
        & 2 & [(1,1) (2,3) +3], [(2,1) (4,2) -2] \\ \hline
    3 & C [(2,0) (4,1) +3], D [(2,1) (4,2) -2], E [(2,2) (3,4) -1]
        & 3 & [(2,2) (3,4) -1] \\ \hline
    4 & - 
        & 4 & [(0,3) (2,4) +2] \\ \hline
\end{tabular}
}
\end{table}

\emph{Supercell update.} To update the pSSs of grid cells for a visited route $R$, the following steps are performed:
\begin{figure}[!htb]
    \centering
    \includegraphics[width=0.6\columnwidth]{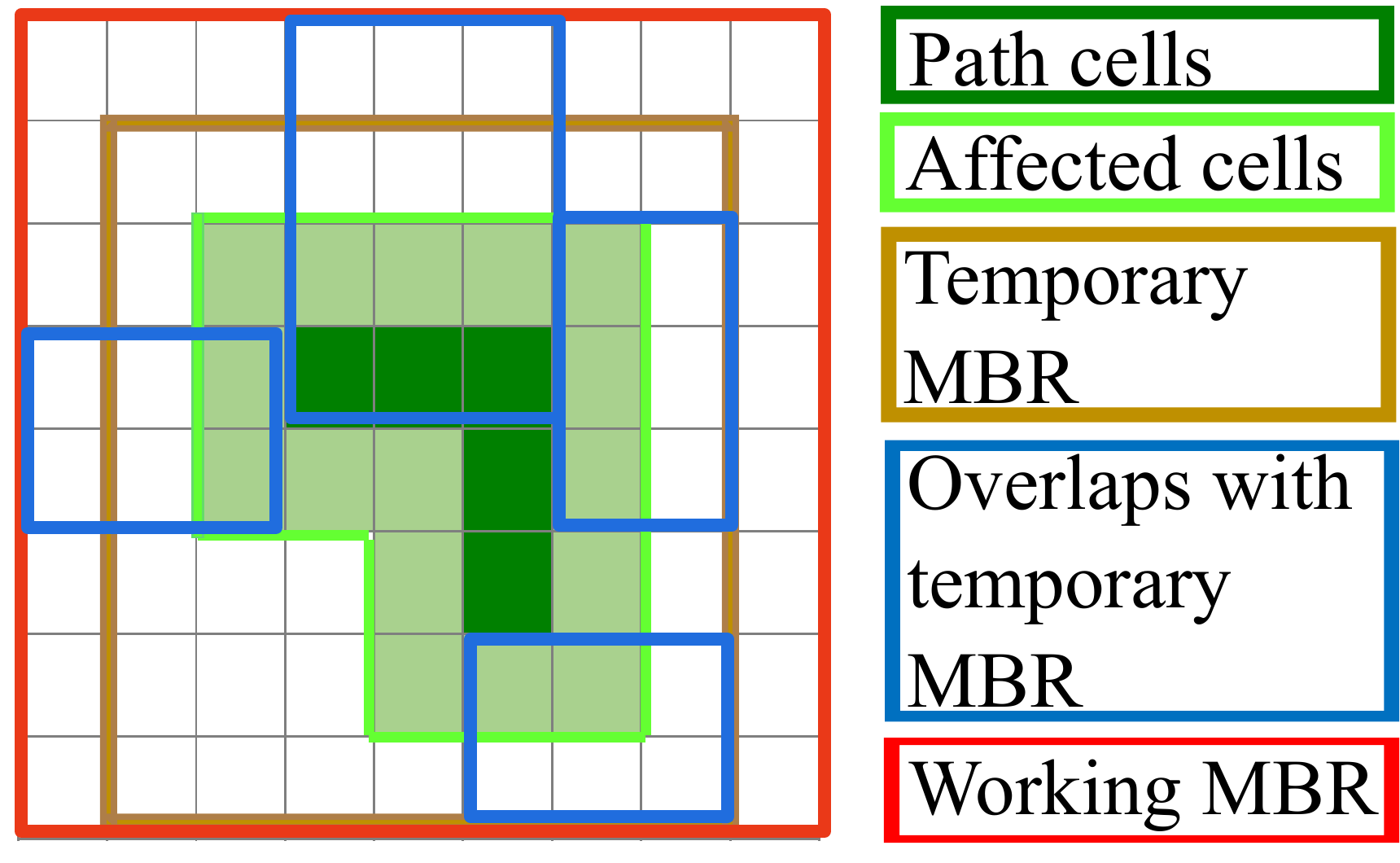}
    \caption{Necessary MBRs for updating a supercell}
    \label{fig:insert-in-rtree}
\end{figure}
\begin{compactitem}
\item \emph{Compute route cells and affected cells.} Compute the grid cells that overlap with $R$ as route cells. The affected cells include the route cells and their nearby cells (Fig.~\ref{fig:insert-in-rtree}).
\item \emph{Compute temporary MBR.} Find the temporary MBR that includes the affected cells and one extra grid cell besides each affected cell in the boundary (Fig.~\ref{fig:insert-in-rtree}). The reason behind considering an extra grid cell is to identify the adjacent existing supercells later.
\item \emph{Find overlapping supercells.} Find existing supercells that intersect with the temporary MBR. There are four overlapping supercells in Fig.~\ref{fig:insert-in-rtree}.
\item \emph{Compute working MBR.} Find the working MBR that includes those overlapping supercells and the affected cells (Fig.~\ref{fig:insert-in-rtree}).
\item \emph{Generate new supercells.} By considering the location, the pSS and the last update time of the grid cells included in the working MBR, generate the new supercells. 
\item \emph{Update $R$-tree.} Remove those overlapping supercells from $R$-tree and add the new supercells. Update the intermediary nodes based on the change in the leaf nodes.
\end{compactitem}
\nc{Fig.~\ref{fig:local-grid-update} shows the updated $R$-tree for the change of the pSS from 3 to -2 in a grid cell (shown with a red circle). Here, the temporary MBR consists of [(0,0) (3,3)] because of including one extra grid cell in each side. The overlapping supercells B, C, D, and E are shown in Fig.~\ref{fig:local-grid}. 
Hence, the working MBR consists of [(0,0) (4,4)]. After that, the computed new supercells A, C', D', E, and F are shown in Fig.~\ref{fig:local-grid-update}. Finally, the overlapping supercells are deleted and new supercells are inserted, which resulted in the $R$-tree of Fig.~\ref{fig:local-grid-update}. }

\nc{
\subsubsection{Complexity Analysis}
The worst case time complexity of supercell computation is $O({x_c}{y_c})$, where $x_c$ and $y_c$ are the number of rows and columns, respectively. If our modified $R$-tree contains $\mathcal{N}$ entries and $\mathcal{N}_{sc}$ new supercells are generated, then the worst case time complexity of $R$-tree update is $O(\mathcal{N}+x_c y_c+\mathcal{N}_{sc}\log \mathcal{N})$, where the working MBR has $x_c$ rows and $y_c$ columns.
The time complexities of the search and delete operations of our modified $R$-tree are the same as the traditional one, which is  $O(\mathcal{N})$ in the worst case.
}

\subsection{Centralized Indexing.} 
The KSs are accessed when the query-relevant groups are computed and updated when a user visits a new area. Since the probability is high that at least a user knows a grid cell area, we store each grid cell's data in a hash map with the grid cell's coordinates as key. For each grid cell, we store the user ids whose KS is 1 for the corresponding grid cell area.

\section{Query Evaluation} \label{our-approach}

\oc{In this section,  we present our query evaluation algorithms. For ease of understanding, we first discuss our algorithms to process the queries to find the SRs between two locations (Section~\ref{sec:SR}). Then we elaborate the generalized algorithms for our safe route planner that can address the SR query and its variants: the FSR query, the GSR query, and the GFSR query (Section~\ref{sec:GFSR}). In Section~\ref{sec:CL}}, we present our measure to assess the reliability of the SR query and its variants. 

\subsection{Evaluation of SR Queries} \label{sec:SR}
In our system, a query requestor retrieves the required pSSs from relevant users and evaluates the SR query. We develop direct and iterative algorithms to find the SR for a source-destination pair $s$ and $d$ within a distance constraint $\delta$. 

The number of possible routes between a source-destination pair can be huge. Retrieving the pSSs for all grid cells that intersect the edges of all possible routes and then identifying the SR would be prohibitively expensive. Our algorithms refine the search space and avoid exploring all routes for finding the SR. We present two optimal algorithms: \emph{\U{Dir}ect \U{O}ptimal \U{A}lgorithm (\dir{})} and \emph{\U{I}terative \U{O}ptimal \U{A}lgorithm} (\itr{}). \dir{} aims at reducing the processing time, whereas \itr{} increases the privacy in terms of the number of retrieved pSSs. Though a pSS does not reveal a user's travel experience (Section~\ref{privacy-analysis}) with certainty, the user's privacy is further enhanced by minimizing the number of shared pSSs with the query requestor. 

\emph{Query-relevant area $A_q$.} Our algorithms exploit the elliptical and Euclidean distance properties to find the query-relevant area $A_q$. We refine the search area using an ellipse where the foci are at $s$ and $d$ of a query and the length of the major axis equals $\delta$. According to the elliptical property, the summation of the Euclidean distances of a location outside the ellipse from two foci is greater than the length of the major axis. On the other hand, the road network distance between two locations is greater than or equal to their Euclidean distance. Thus, the road network distance between two foci, i.e., $s$ and $d$ through a location outside the ellipse, is greater than $\delta$. The refined search area $A_q$ includes the grid cells that intersect with the ellipse. $A_q$ enables us to select a query-relevant group and mitigate unnecessary processing and communication overheads and data exposure.

\emph{Query-relevant group $G_q$.}
A query-relevant group $G_q$ consists of the users whose KS is 1 for at least one grid cell in $A_q$. After receiving a query, the centralized server sends $G_q$ and the list $M_q$ of knowledgeable group members for every grid cell in $A_q$ to the query requestor. 

\begin{small}
\begin{algorithm}[!htb]
    
    
        
        

$N^\prime, A_q \gets$ compute\_query\_area$(s, d, \delta, N)$\;
$G_q, M_q \gets$ retrieve\_query\_group$(A_q)$\;
$SS_q \gets$ compute\_SS$(G_q, M_q, A_q)$\;
$N^{\prime\prime} \gets$ refine\_query\_area$(s, d, \delta, N^\prime, SS_q)$\;
$SR \gets$ compute\_safest\_route$(s, d, \delta, N^{\prime\prime}, SS_q)$\;
\Return $SR$\;

    \caption{\dir{}($s$, $d$, $\delta$, $N$)}
    \label{alg:dirOA}
\end{algorithm}
\end{small}

\subsubsection{Direct Algorithm (\dir{})} \label{sec:dirOA}

One may argue that we can simply apply an efficient shortest route algorithm (e.g., Dijkstra) for finding the SR by considering the SS instead of the distance as the optimizing criteria. However, it is not possible because the SR identified in this way in most of the cases may exceed $\delta$. 

Algorithm~\ref{alg:dirOA} shows the pseudocode for \dir{}. The algorithm starts by computing the query-relevant area $A_q$ and the query-relevant road network $N^\prime$ that is included in $A_q$. The edges in $N$ that go through grid cells in $A_q$ but those cells have not been visited by any user are not included in $N^\prime$. Then the algorithm retrieves the query-relevant group $G_q$ and the map $M_q$ of grid cell-wise knowledgeable group members from the centralized server. In the next step, the algorithm retrieves the pSSs from the group members and aggregates them to compute the SSs of the grid cell in $A_q$ using Function compute\_SS. 

After having the SSs for the grid cells in $A_q$, the algorithm further refines $N^\prime$ to $N^{\prime\prime}$ by pruning the edges that are guaranteed to be not part of the SR (Line 4). The idea of this pruning comes from~\cite{DBLP:journals/is2016/Urban-navigation}, where edges with the lowest SSs are incrementally removed until $s$ and $d$ become disconnected. To reduce the processing time, we exploit binary search for finding $N^{\prime\prime}$. Specifically, we compute the mid value $mid$ of the lowest and the highest SSs, i.e., $-S$ and $S$, and remove all edges that have SS lower than or equal to $mid$. Note that an edge can have more than one associated SSs as it can go through multiple grid cells. For binary search, we consider the minimum of these SSs as the SS of the edge. After removing the edges, we find the shortest route between $s$ and $d$ and check if the length of the shortest route satisfies $\delta$. If no such route exists, then the removed edges are again returned to $N^{\prime\prime}$, and the process is repeated by setting the highest SS to $mid$. On the other hand, if such a route exists, the process is repeated by setting the lowest SS to $mid+1$. The repetition of the process ends when the lowest SS exceeds the highest one. 

Finally, \dir{} searches for the SR within $\delta$ in $N^{\prime\prime}$ using Function compute\_safest\_route. \dir{} starts the search from $s$ and continuously expands it through the edges in the road network graph $N^{\prime\prime}$ until the SR is identified. The algorithm keeps track of all routes instead of the safest one from $s$ to other vertices in $N^{\prime\prime}$ as it may happen that expanding the SR from $s$ exceeds $\delta$ before reaching $d$. 

The compute\_safest\_route function uses a priority queue $Q_p$ to perform the search. Each entry of $Q_p$ includes a route starting from $s$, the road network distance of the route, and the distance associated with each SS in the route. The entries in $Q_p$ are ordered based on the safety rank, i.e., the top entry includes the SR among all entries in $Q_p$. Initially, routes are formed by considering each outgoing edge of $s$. Then the routes are enqueued to $Q_p$. Next, a route is dequeued from $Q_p$ and expanded by adding the outgoing edges of the last vertex of the dequeued route. The formed routes are again enqueued to $Q_p$. The search continues until the last vertex of the dequeued route is $d$.     
While expanding the search we prune a route if it meets any of the following two conditions: 

\begin{compactitem}
\item If the summation of the road network distance of the route and the Euclidean distance between the last vertex of the route and $d$ exceeds $\delta$.  
\item If the road network distance of the route exceeds the current shortest route distance of the last vertex from $s$. 
\end{compactitem}

Both pruning criteria guarantee that the pruned route is not required to expand for finding the SR. The current shortest route in the second pruning condition for a vertex $v$ from $s$ is determined based on the distances of the dequeued routes whose last vertex is $v$. Since the dequeued routes to $v$ are safer than a route that has not been enqueued yet, the route can be safely pruned if its length is greater than the current shortest route's distance.  

\begin{small}
\begin{algorithm}[!htb]
    $N^\prime, A_q \gets$ compute\_query\_area$(s, d, \delta, N)$\;
$G_q, M_q \gets$ retrieve\_query\_group$(A_q)$\;
$SS_q \gets \emptyset$, $Q_p \gets \emptyset$, $v \gets s$\;
\While{$v!=d$}{
${A_q}^\prime \gets$ find\_required\_cells$(v, N^\prime, A_q, SS_q)$\;
$SS_q \gets SS_q \bigcup$ compute\_SS$(G_q, M_q, {A_q}^\prime)$\;
$SR \gets$ get\_safest\_route$(v, N^\prime, SS_q, Q_p)$\;
$v \gets$ get\_last\_vertex$(SR)$\;
}
\Return $SR$\;

    \caption{\itr{}($s$, $d$, $\delta$, $N$)}
    \label{alg:itOA}
\end{algorithm}
\end{small}

\subsubsection{Iterative Algorithm (\itr{})}\label{sec:itOA}
\itr{} enhances user privacy by reducing the shared pSSs with the query requestor as it does not need to know the SSs of all grid cells in $A_q$. Algorithm~\ref{alg:itOA} shows the pseudocode for \itr{}. Similar to \dir{}, \itr{} computes $N^\prime$, $A_q$, $G_q$, and $M_q$. \itr{} does not apply the binary search to further refine $N^\prime$ as it avoids retrieving the pSSs of all grid cells in $A_q$. \itr{} gradually retrieves the pSSs from the group members only for the grid cells that are required for finding SR. Another advantage of \itr{} is that it only involves those group members who know about the required grid cells.

\itr{} iteratively searches for the SR in $N^\prime$ using a priority queue $Q_p$ like \dir{}. \itr{} expands the search by exploring the outgoing edges of $v$. Initially $v$ is $s$ and later $v$ represents the last vertex of the dequeued route from $Q_p$. In each iteration, \itr{} identifies the grid cells in ${A_q}^\prime$ through which those outgoing edges pass (Function find\_required\_cells) and computes their SSs by retrieving pSSs from the group members (Function compute\_SS). 
Next, inside Function get\_safest\_route, \itr{} forms the new routes by adding the outgoing edges of $v$ at the end of the last dequeued route and enqueues them into $Q_p$ if they are not pruned using the conditions stated for \dir{}. In the end, the function dequeues a route from $Q_p$ for using in the next iteration. The search for SR ends if the last vertex of the dequeued route is $d$.

As expected, \gitr{} increases the communication frequency (comm. freq.) of the query requestor with the query-relevant group members. \nc{We use a parameter $X_{it}$ such that 
when  $X_{it}>1$, Function find\_required\_cells identifies the grid cells of the outgoing edges of $v$ up to depth $X_{it}$ while applying the first and second pruning techniques, and Function compute\_SS collects the pSSs of those grid cells at once. This way $X_{it}$ decreases the comm. freq. by slightly increasing the number of pSSs retrieved. Note that the algorithm does not collect the pSSs of the edges to be expanded if their SSs are already known (i.e. their retrieved pSSs have been retrieved in the previous iteration).
}

\subsubsection{Complexity Analysis} \label{complexity-sr}
\oc{The compute\_safest\_route function in \dir{} algorithm can be drawn as a tree where the source node is the root and the destination node is in the last level. If the average branching factor is $b$ and the average depth of a route from $s$ to $d$ is $p$, then $Q_p$ is dequeued $1+b+b^2+\ldots+b^p$ times. The maximum possible number of elements at a time in $Q_p$ is $O(b^p)$. Therefore, if a binary min-heap is used for $Q_p$, then the runtime complexity of \dir{} is $O(b^p\cdot \log(b^p))$. Since we utilize two pruning techniques due to which the average depth $p$ reduces to $\frac{p}{r}$, the complexity becomes $O(b^\frac{p}{r} \cdot \log (b^\frac{p}{r}))$.} 

\oc{In \itr{}, edges are expanded till depth $X_{it}$ in Function find\_required\_cells along with the two pruning techniques. Therefore, the runtime complexity of this function is $O(b^\frac{X_{it}}{r})$. Thus, the runtime complexity of \itr{} is  $O(b^\frac{p}{r} \cdot ( \log (b^\frac{p}{r}) + b^\frac{X_{it}}{r} ) )$. 
}

\subsection{Safe Route Planners}\label{sec:GFSR}

\oc{To evaluate the SR query variants: the FSR query, the GSR query, and the GFSR query, we can independently find the SR between every source-destination pair by applying \dir{} or \itr{} and then select the set of routes that maximize road safety within the distance constraint. However, this na\"ive solution would require the traversal of the same road network multiple times and incur excessive processing overhead. We denote these na\"ive approaches as \emph{\U{N}a\"ive \U{Dir}ect \U{A}lgorithm} (\ndir{}) and \emph{\U{N}a\"ive \U{It}erative \U{A}lgorithm} (\nitr{}), respectively. To overcome the limitation of the na\"ive approaches, in this section, we present \emph{\U{G}eneralized \U{Dir}ect \U{A}lgorithm} (\gdir{}) and \emph{\U{G}eneralized \U{It}erative \U{A}lgorithm} (\gitr{}), that generalize \dir{} and \itr{}, respectively, and can find the answer of an SR query and its variants with a single search in the road network.} 

\begin{figure}
    \centering
    \includegraphics[width=0.5\textwidth]{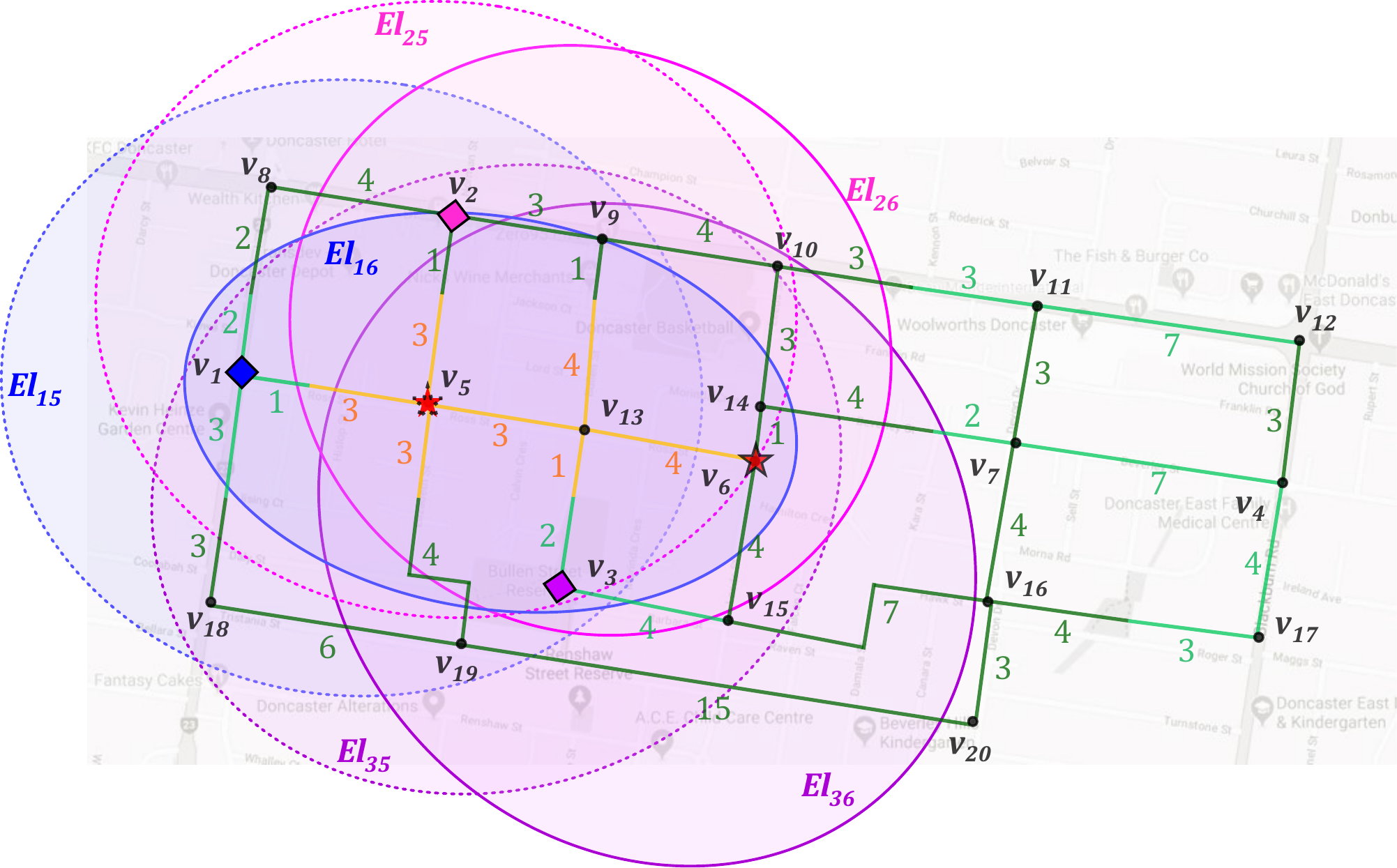}
    \caption{The query-relevant area comprises of ellipses $El_{15}$, $El_{16}$, $El_{25}$, $El_{26}$,  $El_{35}$ and $El_{36}$, where $L_S=\{v_1,v_2,v_3\}$, $L_D=\{v_5,v_6\}$ and $\delta=13$ units.}
    \label{fig:query-area}
\end{figure}

\emph{Query-relevant area and query-relevant group.}
\oc{Finding the answers of SR query variants with a single road network search is quite challenging, as the search space can be extremely large, especially for $n$ sources and $m$ destinations. By utilizing the distance constraint, we refine the search space and find the query-relevant area. For this purpose, for each pair of $s_i$ and $d_j$ ($s_i\in L_S$, $d_j\in L_D$), we find an ellipse whose foci are at $s_i$ and $d_j$ and the major axis is equal to $\delta$. There are $n\cdot m$ such ellipses. Therefore, our query-relevant area $A_q$ is the grid cells that intersect with any of those $n\cdot m$ ellipses. Fig.~\ref{fig:query-area} shows an example of the query-relevant area for a query with $L_S=\{v_1,v_2,v_3\}$, $L_D=\{v_5,v_6\}$ and $\delta=13$. Here, the road network area inside $A_q$ is $N^\prime$, which includes edges that go through grid cells in $A_q$ that have been visited by at least one user. The query-relevant group $G_q$ is formed by including any user whose knowledge score is 1 in at least one grid cell of $A_q$. }

\subsubsection{Generalized Direct Algorithm (\gdir{})} 
\oc{\gdir{}, a generalization of \dir{}, collects necessary pSSs at once, and efficiently finds the solution for the SR query and its variants. However, it is quite challenging to extend \dir{} for multiple source and/or destination locations and find the safest set of routes, $\mathcal{H^*}$. In addition, some of the search space refinement techniques of \dir{} are no longer valid here. Therefore, generalizing \dir{} is not trivial.} 

\oc{\gdir{} takes $L_S$ and $L_D$ as input and computes the query-relevant area $A_q$ and network $N^\prime$ and query-relevant group $G_q$ as described above. The binary search based refinement of \dir{} gradually removes the unsafer edges and finds the smallest road network that has the source and destination locations of the SR query connected. On the other hand, for FSR, GSR and GFSR queries when there are multiple source and destination locations, the binary search based refinement needs to ensure that at least one destination $d_j\in L_D$ is reachable from all $s_i\in L_S$ within $\delta$. To check the connectivity between a source and a destination, we can simply find the shortest route between the source-destination pair, which does not apply when we generalize the binary search for the SR query variants because the nearest destination of multiple sources is not known in advance and may not be the same for all source locations. Therefore, in the binary search based refinement of \gdir{}, we start expanding the routes from each $s_i$ and terminate the search once we reach any destination location $d_j\in L_D$ from all source locations $s_i\in L_S$ through routes having distances smaller than or equal to $\delta$.} 

\oc{To summarize, the binary search based refinement of \gdir{} works as follows. We compute the $mid$ value of the SS range $[-S,S]$ and remove all edges from $N^{\prime}$ that have SS lower than or equal to $mid$. After removing those edges, we confirm that at least one destination $d_j\in L_D$ is reachable from all $s_i\in L_S$ within $\delta$. If such a $d_j$ exist, then we repeat the process for SS range $[mid+1,S]$. However, if no such $d_j$ is found, we bring back the removed edges, and update the SS range to $[-S,mid]$. We repeat this refinement technique until the SS range diminishes to a single SS. After this refinement, we get the road network $N^{\prime\prime}$.}

\oc{Finally, we compute our answer $\mathcal{H^*}$ by exploring $N^{\prime\prime}$ efficiently. Though multiple source and destination locations can be present in the query, we search for $\mathcal{H^*}$ in $N^{\prime\prime}$ by efficiently using a single priority queue, $Q_p$. The routes stored in $Q_p$ are ordered based on their safety ranks similar to \dir{}. We start the search for $\mathcal{H^*}$ from all $s_i\in L_S$ simultaneously using $Q_p$. At first, we enqueue outgoing edges of each $s_i\in L_S$ in $Q_p$. Then, we continue dequeuing the SRs from $Q_p$ until we find a destination $d$ such that $d$ has been reached from all $s_i\in L_S$. Specifically, after dequeuing the SR, we form new routes from its outgoing edges and enqueue them. To make the search feasible for real-time query answering, similar to \dir{}, before enqueueing a route we check whether the route can be pruned. We directly apply the second pruning technique of \dir{} and generalize the first pruning technique as follows.
\begin{compactitem}
\item if the summation of the road network distance of a route $R$ and the Euclidean distance between the last vertex of $R$ and $d_j$ exceeds $\delta$ for all $ d_j\in L_D$, then prune the route $R$.
\end{compactitem}
This pruning technique only prunes routes that are not required to compute $H^*$.}

\subsubsection{Generalized Iterative Algorithm (\gitr{}).} 

\begin{small}
\begin{algorithm}[htb!]
    \oc{
$N^\prime, A_q \gets$ compute\_query\_area$(L_S, L_D, \delta, N)$\;
$G_q, M_q \gets$ retrieve\_query\_group$(A_q)$\;
$SS_q \gets \emptyset$, $Q_p \gets \emptyset$\;
$\mathcal{H}_1, \mathcal{H}_2, \ldots, \mathcal{H}_m,\mathcal{H}^* \gets \emptyset$\;
\ForEach{ $s_i \in L_S$ }{
    enqueue($Q_p, \{s_i\}$)\;
}
$SR \gets$ dequeue($Q_p$)\;
$v \gets$ get\_last\_vertex$(SR)$\;
\While{$\mathcal{H}^*=\emptyset$ $and$ $Q_p ~is ~not ~empty$}{
${A_q}^\prime \gets$ find\_required\_cells$(v, N^\prime, A_q, SS_q)$\;
$SS_q \gets SS_q \bigcup$ compute\_SS$(G_q, M_q, {A_q}^\prime)$\;
$SR \gets$ get\_safest\_route$(v, N^\prime, SS_q, Q_p)$\;
$\mathcal{H}_1, \mathcal{H}_2, \ldots, \mathcal{H}_m  \gets$ update\_route\_sets($SR$,$L_D$)\;
$\mathcal{H}^*  \gets$ compute\_answer($\mathcal{H}_1, \mathcal{H}_2, \ldots, \mathcal{H}_m$)\;
$v \gets$ get\_last\_vertex$(SR)$\;
}
\Return $\mathcal{H}^*$\;
}


    \caption{\gitr{}($L_S$, $L_D$, $\delta$, $N$)}
    \label{alg:itEA}
\end{algorithm}
\end{small}

\oc{\gitr{} efficiently finds the solution for the SR query or its variants where the necessary pSSs are collected iteratively. It minimizes the number of shared pSS with the query requestor. Generalizing \itr{} to \gitr{} for evaluating the SR query variants with an aim to further minimize the shared pSSs is not straightforward. The pseudocode of \gitr{} is shown in Algorithm~\ref{alg:itEA}.
}

\oc{In \gitr{}, we first compute the query-relevant area $A_q$ and road network $N^\prime$ (Function compute\_query\_area), and query-relevant group (Function retrieve\_query\_group) as explained above.}

\oc{Let $H_j$ store the SRs from each $s_i\in L_S$ to $d_j$ within $\delta$. $H_j$ is initialized to $\emptyset$ (Line 4). \gitr{} searches for the safest set of routes, $\mathcal{H^*}$, in $N^\prime$ using a single priority queue $Q_p$. As always, the routes in $Q_p$ are ordered based on their safety ranks. Initially, \gitr{} forms a route from each $s_i$ and enqueue it to $Q_p$ (Lines 5-7). Then \gitr{} dequeues a route and finds the last vertex of the dequeued route as $v$. In the Loop of Lines 10-17, \gitr{} expands the search by exploring the outgoing edges of $v$. \nc{Function find\_required\_cells identifies the grid cells in $A_q$ through which those outgoing edges pass and Function compute\_SS computes their SSs by retrieving pSSs from the group members.} 
Then, Function get\_safest\_route forms new routes by including outgoing edges of $v$ at the end of $SR$ and enqueues them in $Q_p$. The function applies the two pruning criteria used in \gdir{} to prune routes before enqueueing them. In the end, Function get\_safest\_route dequeues a route from $Q_p$ and returns it as $SR$.} 

\oc{Next, Function update\_route\_sets adds $SR$ to $H_j$ if $SR$ is the first dequeued route from $s_i$ to $d_j$. Because of the ordering scheme of $Q_p$, the first dequeued route from $s_i$ that reaches $d_j$ is the SR from $s_i$ to $d_j$.} \oc{Finally, if any route-set $H_j$ includes  the SR for all $s_i\in L_S$, Function compute\_answer returns the corresponding SRs as the query answer.} Otherwise, the function returns $\emptyset$. This returned value is set to $\mathcal{H}^*$. The search ends when $Q_p$ is empty or the answer is found.

\oc{Next, Function update\_route\_sets adds $SR$ to $H_j$ if $SR$ is the first dequeued route from $s_i$ to $d_j$. Because of the ordering scheme of $Q_p$, the first dequeued route from $s_i$ that reaches $d_j$ is the SR from $s_i$ to $d_j$.} \nc{Finally, if any route-set $H_j$ includes  the SR for all $s_i\in L_S$, Function compute\_answer returns the corresponding SRs as the query answer.} {Otherwise, the function returns $\emptyset$. This returned value is set to $\mathcal{H}^*$. The search ends when $Q_p$ is empty or the answer is found.}

\oc{As expected, \gitr{} also increases the communication frequency of the query requestor with query-relevant group members. Therefore, similar to \itr{}, we utilize the parameter $X_{it}$ that trades off between the communication frequency and the number of pSSs shared with the query requestor in \gitr{}.}

\subsubsection{Complexity Analysis} \label{complexity-gen}
\nc{For $n$ source locations, $m$ destination locations and using a binary min-heap for $Q_p$, the worst case time complexities of \gdir{} and \gitr{} are $O(n b^\frac{p}{r} ( m + \log(nb^\frac{p}{r}) ) )$ and $O(n b^\frac{p}{r} ( m + \log(nb^\frac{p}{r}) + b^\frac{X_{it}}{r} ) )$, respectively, where $b$ is the average branching factor and the average depth $p$ of a route reduces to $\frac{p}{r}$ due to pruning effect.} 

\subsection{Confidence Level}
\label{sec:CL}
The confidence level of a query answer expresses its reliability from the viewpoint of a query requestor. In our case, the more the number of users supports an answer, the more reliable it is to the query requestor. For a route $R$, its confidence level $CL(R)$ is expressed as follows.
$$CL(R) =\frac{100}{z}\times \frac{\sum_{c_i} l_i \times m_{c_i}}{dist(R) \times m}$$ 

Here, $l_i$ is the length of $R$ that crosses grid cell $c_i$ and $m_{c_i}$ is the number of group members who know $c_i$. Intuitively, the CL indicates the average percentage of query-relevant group members who know each unit length of the route. The query requestor might be satisfied when on average $z\%$ members among the $m$ query-relevant group members know about each unit length. Thus, we include $z$ in the definition of the CL. 

\begin{figure}[!h]
\vspace{-1mm}
\centering
\includegraphics[width=3.5cm]{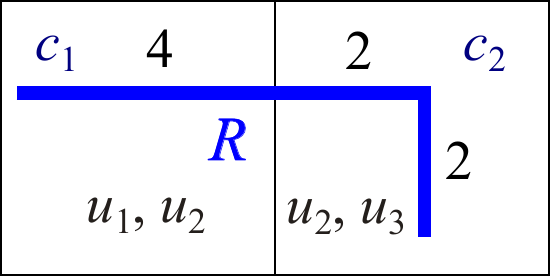}
\caption{\nc{CL calculation}}\label{fig:cl-example}
\vspace{-2mm}
\end{figure}    
\nc{ In Fig.~\ref{fig:cl-example}, a route $R$ goes through $c_1, c_2$, and $l_1=4$, $l_2=4$, $m_{c_1}=2$, $m_{c_2}=2$, $m=3$, and $dist(R)=8$. Thus, $CL(R) = min(\frac{100}{z}\times \frac{4\times 2+4\times 2}{8 \times 3}, 1) = min(\frac{100\times 0.67}{z},1)$. When a query requestor has a high requirement of $z$ to feel confident, it is difficult to satisfy, thus the CL is low, e.g., $CL(R) = min(0.67,1)=0.67$ for $z=100$. When the requirement is low, the CL increases, e.g., $CL(R)=min(2\times 0.67,1)=1$ for $z=50$. }

\nc{The result of an FSR query contains one SR, thus the computation of CL is the same as the SR query. For GSR or GFSR queries, there are $n$ SRs. Thus, we compute the CL for each of them and take the average as the CL of these queries.}
\section{Privacy Analysis} \label{privacy-analysis}

\subsection{Privacy Attacker Model} 
We assume a semi-honest setting, where participants follow the system protocol but are curious to infer unsafe event type (e.g., hijacking, sexual harassment) faced by a user from the shared information. In this setting, the participants do not send queries with the intention of exposing pSSs, they do not send the wrong pSSs, and they do not collude with each other or any centralized entity to infer the private data.


Anyone can play the role of an adversary.  The adversary can have the following knowledge: (i) the model to compute the pSSs and the algorithms to compute the SRs, (ii) the \nc{KSs (knowledge scores)} and pSSs of the query relevant area for the users who participate at the adversary's query evaluation, and (iii) the time and location of \emph{some} of the unsafe events that occur on roads. 

The adversary can know that a user has faced an unsafe event (\emph{not the type}) from a shared negative pSS of a grid cell. The location and time of this inferred unsafe event are also not certain as a grid cell can have a negative pSS due to a past or a recent unsafe event occurring at a nearby grid cell. A KS only discloses whether the user visited an area or not. Since the adversary does not have any knowledge about the frequency and the time of the user's visits (event) in a grid cell, even if the adversary knows that an unsafe event occurs at a grid cell, the user's KS or negative pSS for a grid cell does not provide any clue to associate the unsafe event type with the user.

Our solution protects user privacy by refraining an adversary from learning \textit{any unsafe event type faced by a user at any grid cell}. Since a negative pSS reveals that a user has faced an unsafe event (\emph{not the type}), our solution further enhances user privacy by minimizing the number of revealed pSSs for a user.



\subsection{Privacy Proof} 
Quantification model parameters, such as the impact of event types ($\xi^+$, $\xi^-_1, \xi^-_2, \cdots, \xi^-_c$), decay parameters ($r_d$ and $\Delta_d$), $h$ and $S$ contribute to the computation of pSSs based on the events that a user has encountered (please see Section~\ref{ss:quantification-of-safety}). 
Among these parameters, only $\xi^-_i$ can vary with unsafe event types and others (e.g., $r_d, \Delta_d$) only diminish the impact of $\xi$. If the impact $\xi^-$ of all unsafe event types are the same, unsafe event types cannot be inferred by an adversary from a shared pSS. When there are $c$ ($\geq 2$) different unsafe event impacts, the following lemma shows the condition for hiding a user's unsafe event type from others.

\begin{lemma}\label{lemma_1} 
Given a user's revealed pSS $\Psi$ for a grid cell, 
the unsafe event types that a user encounters cannot be inferred if (i) more than one event combinations cause the model to result in $\Psi$ and (ii) every unsafe event type is not included in at least one event combination that results in $\Psi$.    
\end{lemma}
\begin{proof} 
An adversary cannot identify the actual event combination that results in $\Psi$ from \emph{multiple} event combinations and infer the unsafe event type of the user because the adversary does not know (i) any unsafe event of the user, and (ii) \emph{all} unsafe events that occur on grid cells. Again, if an unsafe event type is not included in at least one event combination that results in $\Psi$, then the adversary cannot infer the user's unsafe event type from $\Psi$.
\end{proof} 
 
Lemma~\ref{lemma_2} shows how our system satisfies Lemma~\ref{lemma_1} for any values of the safety quantification model parameters.

\newcommand{\rd}[1]{r_d^{\floor{\frac{#1}{\Delta_d}}}}
\begin{lemma} \label{lemma_2}
For any values of $\xi^-_i$ ($1 \leq i \leq c$, $c\geq 2$ and $\xi^-_{i+1} < \xi^-_{i}$), $\xi^+$, $r_d$, $\Delta_d$, $h$ and $S$, the conditions of Lemma~\ref{lemma_1} are satisfied for a user's revealed pSS $\Psi$.
\end{lemma}
\begin{proof} 
According to our safety quantification model, one possible way to form $\Psi$ is $\Psi = \floor{ x_{k_0} \cdot \xi^-_k + y_{r} \cdot \xi^+ \rd{r} }$, 
where $x_{k_0}, y_r\in \Z$ and $x_{k_0}$ ($y_r$) represents the number of times an unsafe (safe) event with impact $\xi^-_k$ ($\xi^+$) occurs zero ($r$) days ago, 
by choosing $x_{k_0}$, $r$ and $ y_r$ in the following way: 
$x_{k_0}$ is the smallest integer such that $\Psi - x_{k_0} \cdot \xi^-_k \geq 0$, $r$ is the smallest integer such that  $\Psi^+ = \xi^+ \rd{r} \leq 1$ and $y_r=\ceil{\frac{\Psi - x_{k_0} \cdot \xi^-_k}{\Psi^+}}$. 
For $c\geq 2$, $\Psi$ can be formed from at least $c$ event combinations using $c$ different values for $\xi^-_k$, i.e. $\xi^-_1, \xi^-_2, \cdots, \xi^-_c$, which satisfies the conditions of Lemma~\ref{lemma_1}. 
\end{proof}

An adversary may apply common sense reasoning to prune some event combinations based on the following observation: a user is unlikely to face an unsafe event type (i) more than once per day, and (ii) daily for a long time. In Lemma~\ref{lemma_4}, we configure $\xi^-_i (1\leq i\leq c)$ and $S$, and show that our solution guarantees Lemma~\ref{lemma_1} by not considering the event combinations under the above observation. 


\newcommand{\mPsi}[1]{\floor{\frac{\xi^-_{#1}(1-r_d^{\mu_{#1}})}{1-r_d}}}
\newcommand{\MPsi}[1]{\floor{\frac{ \xi^-_{#1} \Delta_d ( 1-r_d^{ \floor{ \frac{\mu_{#1}-1}{\Delta_d}}  } )  }{  1-r_d  }}}
\newcommand{\extra}[1] {(\floor{ \frac{#1}{\Delta_d}+1})\Delta_d -1-{#1} }

\begin{lemma} \label{lemma_4} 
For any values of $\xi^+$, $\xi^-_i$ ($1 \leq i \leq c$, $c\geq 2$ and $\xi^-_{i+1} < \xi^-_{i}$), $r_d$, $\Delta_d$ and $h$ such that $\xi^-_{i+1} \geq \mPsi{i}$ and $-S \geq \mPsi{c-1}$, the conditions of Lemma~\ref{lemma_1} are satisfied, where $\xi^-_i$ ($\xi^-_{c-1}$) occurs at most once per day for at most $\mu_i$ ($\mu_{c-1}$) consecutive days.
\end{lemma}
\begin{proof} 
If $\Delta_d=1$, 
then the smallest pSS formed from $\xi^-_k$ is $\Psi^{\star}_k =\floor{\xi^-_k (1 + r_d + r_d^2 + \cdots + r_d^{\mu_{k-1}})}=\mPsi{k}$ (Eq. 1), where an event of impact $\xi^-_k$ occurs daily once for $\mu_k$ days. 

\emph{Case $\Psi\geq\Psi^{\star}_k$ ($1\leq k < c$): }
According to our safety quantification model, one possible way to form $\Psi\geq\Psi^{\star}_k$ is $\Psi = \floor{\xi^-_k ( 1 + r_d + \cdots + r_d^{p_k}) + \xi^+ ( y_0\cdot 1 + y_1\cdot r_d + \cdots + y_q\cdot r_d^{q}) } = \floor{\Psi^- + \Psi^+}$ (Eq. 2), where $p_k, q, y_q\in \Z$, $\xi^-_i$ occurred daily since $p_k(<\mu_k)$ days ago, and $y_q$ represents the number of times $\xi^+$ occurred $q$ days ago, by choosing $p_k$, $y_0,\cdots,y_q$ in the following way: $p_k$ ($<\mu_k$) is the smallest integer such that $\Psi^+ = \Psi - \Psi^- \geq 0$, $y_0=\floor{\frac{\Psi^+}{\xi^+\cdot 1}}$, $y_1=\floor{\frac{\Psi^+ -  \xi^+\cdot y_0\cdot 1}{\xi^+ r_d^1}}$, $\cdots$, $y_q=\floor{\frac{\Psi^+ - \xi^+ (y_0\cdot 1 + y_1\cdot r_d^1 + \cdots + y_{q-1}\cdot r_d^{q-1})}{\xi^+ r_d^q}}$, 
and $\floor{\Psi^+ - \xi^+ (y_0\cdot 1 + y_1\cdot r_d + \cdots + y_{q}\cdot r_d^{q})}=0$. 
Here, $\Psi^+=\Psi - \Psi^- \geq 0$ is true even when $\xi^-_k$ is replaced with any impact $\xi^-_i$, $k < i\leq c$. 
Therefore, $\Psi\geq \Psi^{\star}_k$ can be formed from at least $c-k+1 \geq 2$ event combinations using $c-k+1 \geq 2$ different values for $\xi^-_k$ in Eq. 2, i.e., $\xi^-_i, \xi^-_{i+1}, \cdots, \xi^-_c$, which satisfies the conditions of Lemma~\ref{lemma_1}.

\emph{Case $\Psi < \Psi^{\star}_{k}$ ($k\geq c-1$): }
This case does not exist due to our constraint $\Psi \geq -S \geq \mPsi{c-1} = \Psi^{\star}_{c-1}$ following Eq. 1.

If $\Delta_d$ is any value, then the smallest pSS formed from $\xi^-_k$ is $\Psi^{\prime\star}_k \leq \MPsi{k} \leq \Psi^{\star}_k$, which increases the number of possible event combinations for $ \Psi \in [\Psi^{\prime \star}_k, \Psi^{\star}_k)$ while keeping the previous event combinations valid. Thus, the conditions of Lemma~\ref{lemma_1} are satisfied.

If $\xi^-_k$ occurs any number of times in a day for a user, then the smallest pSS $\Psi^{\prime\prime\star}_k$ formed from $\xi^-_k$ will be far smaller than $\Psi^{\star}_k$, which again increases the number of possible event combinations for those $\Psi^{\prime\prime\star}_{k}\leq \Psi< \Psi^{\star}_{k}$. Thus, the conditions of Lemma~\ref{lemma_1} are satisfied.
\end{proof} 

\nc{
Similar to Lemma~\ref{lemma_4}, the conditions of Lemma~\ref{lemma_1} can also be satisfied under other common sense reasonings (e.g., a user is unlikely to face a safe event more than once per day) by  deriving rules for choosing the values of our safety quantification model parameters.
}


\subsection{Privacy Enhancement}
In our system, the following measures further enhance the privacy of user data:
\begin{compactitem}
     \item {Our  solution refines the search space and minimizes the number of shared pSSs with a query requestor. 
     \item A user shares pSSs with the query requestor instead of a centralized server. A centralized server is fixed, and thus, a user would not feel comfortable sharing all pSSs with the centralized server, whereas a query requestor changes with a query, and the user only shares limited query-relevant pSSs with the query requestor.}  
     
     \item {A user can choose not to share her KSs for sensitive areas with the centralized server.}

     \item Our solution does not need to store the event data. It transforms a user's events into pSSs and stores them on the local device. Thus, an adversary can only retrieve a user's pSSs by applying an attack on the local device. 
\end{compactitem}

\color{black}

\makeatletter  
\newcommand\DummyPlot[8]{
\begin{tikzpicture}
\begin{axis}[
    width=8cm,
    height=9cm,
    xlabel style={font=\fontsize{24}{0}\selectfont, },
    ylabel style={font=\fontsize{24}{0}\selectfont},
    xlabel={#1},
    ylabel={#2},
    xtick=data,
    x tick label style={,font=\fontsize{24}{0}\selectfont},
    y tick label style={font=\fontsize{24}{0}\selectfont},
    legend style={at={(0.5,1.15)},anchor=south,legend columns=2,font=\fontsize{24}{0}\selectfont},
    mark options={mark size=5.5pt, line width=1.5pt},
    cycle list name={my six colors},
]
\addplot table [x=#3,y=#4, col sep=comma] {#6};
\addplot table [x=#3,y=#5, col sep=comma] {#6};
\addplot table [x=#3,y=#4, col sep=comma] {#7}; 
\addplot table [x=#3,y=#5, col sep=comma] {#7};
\addplot table [x=#3,y=#4, col sep=comma] {#8}; 
\addplot table [x=#3,y=#5, col sep=comma] {#8};
\end{axis}
\end{tikzpicture}
}
\newenvironment{DummyPlotEnv}{%
\noindent
  \def\addXTitle##1{\def\mXT{##1}}
  \def\addYTitle##1{\def\mYT{{##1}}}
  \def\addX##1{\def\mX{##1}}
  \def\addY##1{\def\mY{##1}}
  \def\addZ##1{\def\mZ{##1}}
  \def\addFP##1##2##3{
    \DummyPlot{\mXT}{\mYT}{\mX}{\mY}{\mZ}{##1}{##2}{##3}
  }
\noindent
}{}
\makeatother


\makeatletter  
\newcommand\SRPlotSix[8]{
\begin{tikzpicture}
\begin{axis}[
    width=8cm,
    height=9cm,
    xlabel style={font=\fontsize{24}{0}\selectfont, },
    ylabel style={font=\fontsize{24}{0}\selectfont},
    xlabel={#1},
    ylabel={#2},
    xtick=data,
    x tick label style={,font=\fontsize{24}{0}\selectfont},
    y tick label style={font=\fontsize{24}{0}\selectfont},
    legend style={at={(0.5,1.15)},anchor=south,legend columns=2,font=\fontsize{24}{0}\selectfont},
    mark options={mark size=5.5pt, line width=1.5pt},
    cycle list name={my six colors},
]
\addplot table [x=#3,y=#4, col sep=comma] {#6};
\addplot table [x=#3,y=#5, col sep=comma] {#6};
\addplot table [x=#3,y=#4, col sep=comma] {#7}; 
\addplot table [x=#3,y=#5, col sep=comma] {#7};
\addplot table [x=#3,y=#4, col sep=comma] {#8}; 
\addplot table [x=#3,y=#5, col sep=comma] {#8};
\end{axis}
\end{tikzpicture}
}
\makeatother
\makeatletter  
\newcommand\SRPlotThree[8]{
\begin{tikzpicture}
\begin{axis}[
    width=8cm,
    height=9cm,
    xlabel style={font=\fontsize{24}{0}\selectfont},
    ylabel style={font=\fontsize{24}{0}\selectfont},
    xlabel={#1},
    ylabel={#2},
    xtick=data,
    x tick label style={,font=\fontsize{24}{0}\selectfont},
    y tick label style={font=\fontsize{24}{0}\selectfont},
    legend style={at={(0.5,1.15)},anchor=south,legend columns=2,font=\fontsize{24}{0}\selectfont},
    mark options={mark size=5.5pt, line width=1.5pt},
    cycle list name={my three colors},
]
\addplot table [x=#3,y=#5, col sep=comma] {#6};
\addplot table [x=#3,y=#5, col sep=comma] {#7}; 
\addplot table [x=#3,y=#5, col sep=comma] {#8}; 
\end{axis}
\end{tikzpicture}
}
\makeatother

\makeatletter  
\newcommand\CPlotFive[4]{
\begin{tikzpicture}
\begin{axis}[
    width=8cm,
    height=9cm,
    xlabel style={font=\fontsize{24}{0}\selectfont, },
    ylabel style={font=\fontsize{24}{0}\selectfont},
    xlabel={#1},
    ylabel={#2},
    xtick=data,
    x tick label style={,font=\fontsize{24}{0}\selectfont},
    y tick label style={font=\fontsize{24}{0}\selectfont},
    legend style={at={(0.5,1.01)}, anchor=south,legend columns=3,font=\fontsize{22}{0}\selectfont},
    legend cell align={left},
    mark options={mark size=5.5pt, line width=1.5pt},
    cycle list name=my accuracy five colors,
] 
\addplot table [x=#3,y=r50, col sep=comma] {#4};
\addplot table [x=#3,y=r60, col sep=comma] {#4};
\addplot table [x=#3,y=r70, col sep=comma] {#4}; 
\addplot table [x=#3,y=r80, col sep=comma] {#4};
\addplot table [x=#3,y=r90, col sep=comma] {#4}; 
\end{axis}
\end{tikzpicture}
}
\makeatother

\makeatletter  
\newcommand\CPlotSix[4]{
\begin{tikzpicture}
\begin{axis}[
    width=8cm,
    height=9cm,
    xlabel style={font=\fontsize{24}{0}\selectfont, },
    ylabel style={font=\fontsize{24}{0}\selectfont},
    xlabel={#1},
    ylabel={#2},
    xtick=data,
    x tick label style={,font=\fontsize{24}{0}\selectfont},
    y tick label style={font=\fontsize{24}{0}\selectfont},
    legend style={at={(0.5,1.01)}, anchor=south,legend columns=3,font=\fontsize{22}{0}\selectfont},
    legend cell align={left},
    mark options={mark size=5.5pt, line width=1.5pt},
    cycle list name=my cl six colors,
]
\addplot table [x=#3,y=r50, col sep=comma] {#4};
\addplot table [x=#3,y=r60, col sep=comma] {#4};
\addplot table [x=#3,y=r70, col sep=comma] {#4}; 
\addplot table [x=#3,y=r80, col sep=comma] {#4};
\addplot table [x=#3,y=r90, col sep=comma] {#4};
\addplot table [x=#3,y=dir, col sep=comma] {#4}; 
\end{axis}
\end{tikzpicture}
}
\makeatother


\makeatletter 
\newcommand\FSRPlotPSS[4]{
\begin{tikzpicture}
\begin{axis}[
    width=8cm,
    height=8cm,
    xlabel style={font=\fontsize{26}{0}\selectfont,yshift=1ex},
    ylabel style={font=\fontsize{26}{0}\selectfont,xshift=-2ex},
    xlabel={#1},
    ylabel={#2},
    xtick=data,
    x tick label style={font=\fontsize{26}{0}\selectfont, , name=xlabel\ticknum},
    y tick label style={font=\fontsize{26}{0}\selectfont},
    legend style={at={(0.5,1.01)}, anchor=south,font=\fontsize{26}{0}\selectfont},
    legend cell align={left},
    mark options={mark size=5.5pt, line width=1.5pt},
    cycle list name=my three fsr colors,
]
\addplot table [x=#3,y=nSSND,col sep=comma] {#4};
\addplot table [x=#3,y=nSSNI,col sep=comma] {#4};
\addplot table [x=#3,y=nSSI,col sep=comma] {#4};
\end{axis}
\end{tikzpicture}
}
\makeatother
\makeatletter 
\newcommand\FSRPlotCF[4]{
\begin{tikzpicture}
\begin{axis}[
    width=8cm,
    height=8cm,
    xlabel style={font=\fontsize{26}{0}\selectfont,yshift=1ex},
    ylabel style={font=\fontsize{26}{0}\selectfont},
    xlabel={#1},
    ylabel={#2},
    xtick=data,
    x tick label style={font=\fontsize{26}{0}\selectfont, , name=xlabel\ticknum},
    y tick label style={font=\fontsize{26}{0}\selectfont},
    legend style={at={(0.5,1.01)}, anchor=south,font=\fontsize{26}{0}\selectfont},
    legend cell align={left},
    mark options={mark size=5.5pt, line width=1.5pt},
    cycle list name=my two fsr colors,
]
\addplot table [x=#3,y=cfNI,col sep=comma] {#4};
\addplot table [x=#3,y=cfI,col sep=comma] {#4};
\end{axis}
\end{tikzpicture}
}
\makeatother
\makeatletter 
\newcommand\FSRPlotRT[4]{
\begin{tikzpicture}
\begin{axis}[
    width=8cm,
    height=8cm,
    xlabel style={font=\fontsize{26}{0}\selectfont,yshift=1ex},
    ylabel style={font=\fontsize{26}{0}\selectfont},
    xlabel={#1},
    ylabel={#2},
    xtick=data,
    x tick label style={font=\fontsize{26}{0}\selectfont, , name=xlabel\ticknum},
    y tick label style={font=\fontsize{26}{0}\selectfont},
    legend style={at={(0.5,1.01)}, anchor=south,font=\fontsize{26}{0}\selectfont},
    legend cell align={left},
    mark options={mark size=5.5pt, line width=1.5pt},
    cycle list name=my four fsr colors,
]
\addplot table [x=#3,y=durND,col sep=comma] {#4};
\addplot table [x=#3,y=durNI,col sep=comma] {#4};
\addplot table [x=#3,y=durD,col sep=comma] {#4};
\addplot table [x=#3,y=durI,col sep=comma] {#4};
\end{axis}
\end{tikzpicture}
}
\makeatother
\section{Experiment}
\oc{We evaluate our safe route planner on real datasets with experiments. First, we discuss the experiment setup in Section~\ref{exp:setup}. Then, we analyze the performance of our safe route planner in Section~\ref{exp:perform}. Finally, in Sections~\ref{exp:missing} and~\ref{exp:query-effect}, we assess the impact of missing data on finding the answer of the SR query and variants, and the effectiveness of SRs, respectively.}

\subsection{Experiment setup}\label{exp:setup}

\subsubsection{Datasets}\label{exp:dataset}
We use datasets of three cities: Chicago (C), Philadelphia (P), and Beijing (B). To simulate the environment, for each dataset, we need the road network data, the crime data, and the users' visit data to different areas. The details of these  datasets are summarized in Table~\ref{tab:dataset-details}. We use datasets of three cities to show the performance of our solution irrespective of the variation in the number of users, check-in and crime data.

We use OpenStreetMap~\cite{r:OpenStreetMap} to download the road networks. We use the real crime data of Chicago~\cite{ChicagoCrimeDataset} and  Philadelphia~\cite{PhiladelphiaCrimeDataset}. For Beijing, instead of crime data, only the locations of crime hotspots' centers~\cite{BeijingCrimeNews} are available. We identify the crime hotspots of Chicago using k-means clustering~\cite{agarwal2013crime} and create hotspots of similar sizes around the Beijing hotspots' centers. We generate daily crimes around the hotspots of Beijing  following the distribution of Chicago. 

\begin{small}
\begin{table}[hbt!]
    \centering
    \addtolength{\leftskip} {-5mm} 
    \addtolength{\rightskip}{-5mm}
    \caption{Datasets}
    \label{tab:dataset-details}
    \begin{tabular}{|p{1.2cm}|R{0.9cm}|R{1.1cm}|R{1.1cm}|R{1cm}|R{1cm}|}
        \hline
        \multirow{2}{1.3cm}{Dataset (30 days)} & \multirow{2}{0.9cm}{\#Users} & \multirow{2}{1.1cm}{\#Check-ins} & \multirow{2}{1.1cm}{\#Crimes} & \multicolumn{2}{c|}{Road Network} \\ \cline{5-6}
            & & & & \#Nodes & \#Edges \\ \hline
        Chicago (\textbf{C}) & 3554 & 60922 & 30843 & 28468 & 74751 \\\hline
        Philadel-phia (\textbf{P}) & 2275 & 26923 & 82363 & 24800 & 59987 \\ \hline
        Beijing (\textbf{B}) & 87 & - & - & 33923 & 75131 \\ \hline
    \end{tabular}
\end{table}
\end{small}

For the users' visit data to different areas, we use the day-to-day Foursquare check-in dataset~\cite{r:checkin1,r:checkin2} for Chicago and Philadelphia, and real trajectory data of users for Beijing~\cite{zheng2011geolife}. We use crime and check-in data of the same 6 months for Chicago and Philadelphia and one-year trajectory data of 87 users for Beijing. \oc{Since all crime events are normally not reported, to increase the number of events, we map these data to one month.}

From check-in data, we generate the users' visits for Chicago and Philadelphia. Specifically, we take two consecutive check-ins of a user in a day and generate an elliptical area, where the foci of the ellipse are located at the check-in locations and the length of the major axis equals 1.25 times the distance between two check-in locations. We consider that the user visited the grid cells in the elliptical area. On the other hand, the user trajectories in Beijing directly provide the grid cell area visited by the users. Since most of the trajectory data is located around the center of Beijing city, we consider the area ([39.7, 40.12, 116.1, 116.6]) around the center of Beijing for our experiments.

We normalize the crime count in the range [0,1] per grid cell for each day. This count represents the crime probability of each grid cell. For each grid cell, according to the crime probability, we randomly associate the crime events with the visits of the users. Thus, the probability of experiencing crime in a grid cell increases for a user who visits the cell multiple times. The visits of the users that are not associated with any crime are considered as safe events. The pSSs are calculated based on the model described in Section~\ref{ss:quantification-of-safety}. We choose the model parameters in a way that satisfies Lemma~\ref{lemma_1} for every pSS.

\oc{To realistically generate the set of flexible destinations for an FSR query or a GFSR query, we need POI-location data. We use the Foursquare Check-in dataset~\cite{r:checkin1,r:checkin2} to get the POI locations of Chicago and Philadelphia. For Beijing, we synthetically generate the POI locations following the uniform distribution which is likely to simulate real POIs~\cite{DBLP:journals/pvldb/AbeywickramaCT16} and we maintain the ratio of the number of POIs vs. the number of road network vertices similar to Chicago. }

\subsubsection{Queries}\label{exp:queries}
\oc{For each experiment, we generate 100 queries randomly and take their average performance. 
To generate an SR query, we randomly choose a source location from the nodes of the road network, and then, choose the farthest node in the road network within query distance, $d_q$ as the destination.
To generate an FSR query, we randomly choose the source location from the nodes of the road network and choose $m$ farthest POIs in the road network within query distance, $d_q$, as destinations. 
We generate a GSR query within a square query area where each side of the square covers $A_q^r$ percent of $d_G$ grid cells. We randomly choose four source locations, one from each side of the square area, from the vertices of the road network. Other $n-4$ source locations are chosen randomly from inside the query area. The destination location is the group nearest POI of the $n$ source locations.
We generate a GFSR query in the same way as a GSR query. The only difference is that $m$ destination locations are the $m$ group nearest POIs of the $n$ source locations.
}

\subsubsection{Parameters}\label{exp:params}
We show the parameters' default values and ranges in Table~\ref{tab:params}. \nc{We divide the total space into $d_G \times d_G$ grid cells. 
} We vary $d_G$ to show the impact of the grid resolution (and the grid cell area) on our solution performance. The distance constraint ratio $\delta_R$ represents the ratio of the allowed road network distance of the safest route and the road network distance of the shortest route from $s$ to $d$. Distance constraint $\delta$ is derived from $\delta_R$ during the query processing time. We keep $\delta_R$ at most 1.5 as a user may not feel comfortable traveling longer than 1.5 times of the shortest distance. The parameter $z$ is used for confidence level (Section~\ref{confidence-level}).
\oc{We vary the number of source locations, $n$, within \{5, 10, 15, 20\}, and choose $n=10$ as the default value, which reflects practical group meetup scenarios. Setting the number of destinations, $m$, greater than 25 will not be realistic. Therefore, we vary $m$ within \{5, 10, 15, 20, 25\}, and the default value is set to 15.} Similar to~\cite{DBLP:journals/is2016/Urban-navigation}, we vary the query distance $d_q$, the Euclidean Distance between $s$ and $d$, from 1 to 5. \oc{We vary the query area parameter within \{5, 10, 15, 20\} and the default value is set to 10.}
 \begin{table}[htb]
    \caption{Parameter settings}
    \label{tab:params}
    \centering
\begin{tabular}{|p{3.5cm}|p{2.5cm}|C{0.9cm}|}
    \hline
    Metric & Range & Default \\
    \hline
    Grid Size, $d_G$ & 300, 500,  800  &  500\\
    \hline
    Distance Constraint Ratio, $\delta_R$ & 1.1, 1.2, 1.3, 1.4, 1.5 & 1.2 \\
    \hline
    \oc{Confidence Level Parameter, $z$ (\%)} &  25, 50, 75, 100 &  50 \\
    \hline
    \oc{\#Source Locations, $n$} & 5, 10, 15, 20 & 10 \\
    \hline
    \oc{\#Destination Locations, $m$} &  5, 10, 15, 20, 25 &  15 \\
    \hline
    Query Distance, $d_q$ (km) & 1, 2, 3, 4, 5 & 5\\
    \hline
    \oc{Query Area Parameter, $A_q^r$ (\%)} &  5, 10, 15, 20 &  10 \\
    \hline
\end{tabular}
\end{table}

For each experiment, we set $S=10$ because a smaller $S$ does not capture the variation of safety, and a large $S$ increases the computation cost by adding insignificant detail. \oc{Thus the pSS range becomes [-10,10].} Our system is written in Java. We run our experiments on an Intel Core i7-7770U 3.60 GHz CPU and 16GB RAM machine.

\subsection{Performance Analysis} \label{exp:perform}
\oc{ Since there is no solution that can find the SRs \emph{in our problem setting} (please see Section~\ref{related-works}), we evaluate the performance of our query processing algorithms by varying a wide range of parameters.} 

We compare the algorithms based on runtime, communication frequency per involved group member (\textit{comm. freq.}), and the total number of revealed pSSs. \oc{The runtime of a query includes of the time to calculate the distance constraint $\delta$ from $\delta_R$. We assume the pSSs are retrieved parallelly from the group members. The KSs and pSSs are updated in offline (i.e., when a user’s device is idle), not during query evaluation. Therefore, they do not affect the query response time.} In addition, note that the fewer the number of revealed pSSs, the better the privacy is. 

\oc{In Section~\ref{exp:sr}, we show the performance of our proposed direct and iterative algorithms (i.e., \dir{} and \itr{}, respectively) to evaluate SR queries. In Sections \ref{exp:fsr}, \ref{exp:gsr}, and \ref{exp:gfsr}, we show the performance of generalized direct and iterative algorithms (i.e., \gdir{} and \gitr{}, respectively) to evaluate the SR query variants: FSR, GSR and GFSR queries. There is no difference between \dir{} and \gdir{} (or \itr{} and \gitr{}) when \gdir{} (or \gitr{}) is applied for SR queries, and thus, we do not show the performance of the generalized algorithms for SR queries separately. SR query variants can be also straightforwardly evaluated by applying the SR query processing algorithm independently using the na\"ive algorithms: \ndir{} and \nitr{} (please see Section~\ref{sec:GFSR} for details). In Sections~\ref{exp:fsr}, \ref{exp:gsr} and \ref{exp:gfsr}, we compare our efficient \gdir{} and \ndir{} with \gitr{} and \nitr{}, respectively.} 

\oc{Though \nitr{} applies \itr{} multiple times, in our implementation, we improve the communication frequency of \nitr{} by ensuring that if the pSSs of a grid cell has been collected for a source-destination pair, then the SS of that grid cell is saved and is not collected or computed again. }

\begin{figure}[!hbt]
    \centering
    \input{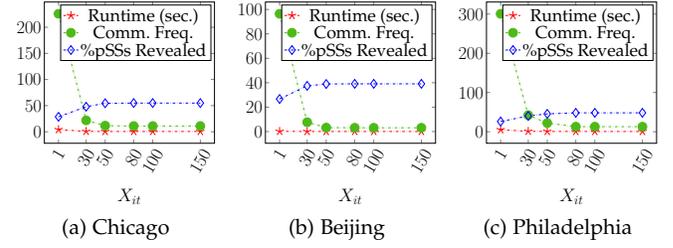}
    \caption{{Choosing default value $X_{it}=40$ based on the effects of $X_{it}$} for SR queries}
    \label{fig:plt:vary-x}
\end{figure}
\noindent\emph{Choosing the default value of $X_{it}$.} 
The parameter $X_{it}$ significantly impacts the performance of \itr{}, \gitr{} and \nitr{}. \oc{Therefore, we vary $X_{it}$ to choose the default value of $X_{it}$ for all query types. In Fig.~\ref{fig:plt:vary-x}, we only plot the results for SR queries.} Fig.~\ref{fig:plt:vary-x} shows clear trade-offs among performance metrics for Chicago. The runtime decreases (desirable), the communication frequency decreases (desirable) and more pSSs are revealed (undesirable) with the increase of $X_{it}$. Thus, we have to carefully choose a value for $X_{it}$ so that the communication frequency is low, and the runtime and the number of revealed pSSs are reasonable. From the figure, it is clear that there is a saturation point after which all three performance metrics do not change much. Therefore, we choose $X_{it}=40$ as the default value for all datasets because we see a sharp decrease in the communication frequency and then for $X_{it}>=40$ there is not much change. Please note that $X_{it}$ can be used as a regulator to control the runtime, communication frequency and data exposure. In scenarios where privacy matters more, we can  choose a lower value for $X_{it}$ and in scenarios where the communication frequency matters, we should choose the value of $X_{it}$ for which communication frequency reaches a saturated value. We decide the value of $X_{it}=40$ for FSR, GSR and GFSR queries in the same way.

\nc{In Section~\ref{exp:rtree}, we show the performance of our modified $R$-tree in terms of its different operations.}

\subsubsection{SR Queries} \label{exp:sr}
\begin{figure}[!bt]
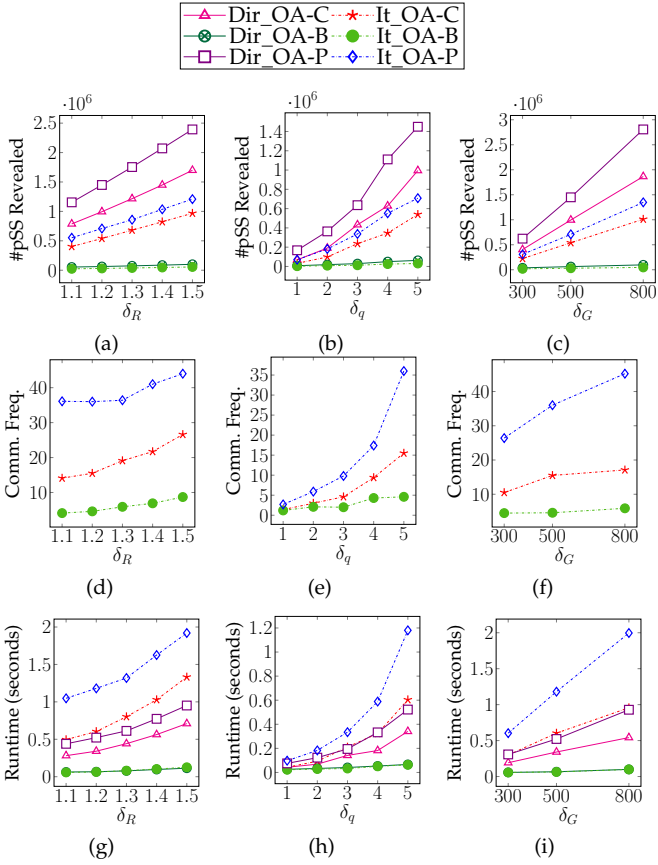

    \pgfplotsset{
    compat=newest,
    /pgfplots/legend image code/.code={%
        \draw[mark repeat=3,mark phase=3,#1] 
            plot coordinates {
                (0cm,0cm) 
                (0.4cm,0cm)
                (0.8cm,0cm)
                (1.2cm,0cm)
                (1.6cm,0cm)%
            };
    },
}
\centerline{\begin{minipage}[t]{0.9\columnwidth}
\centering
\scalebox{0.35}{
\begin{tikzpicture} 
    \begin{axis}[
        hide axis,
        xmin=0,
        xmax=0,
        ymin=-1,
        ymax=-1,
        legend style={at={(0.5,1.15)},anchor=south,legend columns=2,font=\fontsize{24}{0}\selectfont,},
        mark options={mark size=6pt, line width=2pt},
        cycle list name=my six colors,
        ]
        \addplot {0};
        \addlegendentry{\dir{}-C};
        \addplot {0};
        \addlegendentry{\itr{}-C};
        \addplot {0};
        \addlegendentry{\dir{}-B};
        \addplot {0};
        \addlegendentry{\itr{}-B};
        \addplot {0};
        \addlegendentry{\dir{}-P};
        \addplot {0};
        \addlegendentry{\itr{}-P};
    \end{axis}
\end{tikzpicture}
}
\end{minipage}
}
    \input{safepath/plots/qe-pp}
    \vspace{1mm}
    \input{safepath/plots/qe-cf}
    \vspace{1mm}
    \input{safepath/plots/qe-rt}
    \caption{\dir{} vs. \itr{} in terms of privacy (\#pSSs revealed) and computation cost (comm.freq. and runtime) for varying $\delta$, $d_q$ and $d_G$}
    \label{fig:plt:qe}
\end{figure}
\oc{Fig.~\ref{fig:plt:qe} shows the comparison of \dir{} and \itr, where we append the initial letter of the dataset after the algorithm name with a hyphen.}

Fig.~\subref*{plt:pp-cDist}-\subref*{plt:pp-grid} shows that \itr{} reveals around 50\% of the revealed pSSs in \dir. The number of revealed pSSs increases with an increase of $\delta_R$ and $d_q$ because the route length increases. The number of revealed pSSs also increases for large $d_G$ because the number grid cells through which the route passes increases. Please note that the number of revealed pSSs increases more rapidly for \dir{} than that of \itr.

Fig.~\subref*{plt:cf-cDist}-\subref*{plt:rt-grid} compare \dir{} and \itr{} in terms of communication frequency  and runtime. The communication frequency  is always 1 for \dir{} as the group-members are requested once to provide pSSs. For \itr, the communication frequency is on average 15 times; it can be as high as 45.2 times (Philadelphia dataset). To check how reasonable this is, we ran an experiment: a message is sent from one device to another using Firebase Cloud Messaging service and a reply from recipient is received. This is a cycle, and we ran 500 such cycles which took a total of 86641 ms, so on average, 173.28 milliseconds per cycle. Therefore, 45.2 communications take $45.2 \times 173.28$ ms $\approx$ 8 seconds, which is reasonable. The communication frequency  for \itr{} increases with an increase of $\delta$, $d_q$ and $d_G$. For $\delta$ and $d_q$, the reason behind the increased communication frequency  is the increased route length, whereas for $d_G$, the reason is the number of required grid cells to compute the SR increases.

The runtime of \dir{} is very low (on average 0.3 second) for all datasets. \nc{Though the runtime of \gitr{} increases with the increase in $\delta$, $d_q$  and $d_G$, they are reasonable (on average 0.54 seconds).} Therefore, we conclude that both \dir{} and \itr{} provide practical solution for the SR queries and show a trade-off between runtime and privacy.

\subsubsection{FSR Queries}\label{exp:fsr}
\begin{figure}[tbh!]
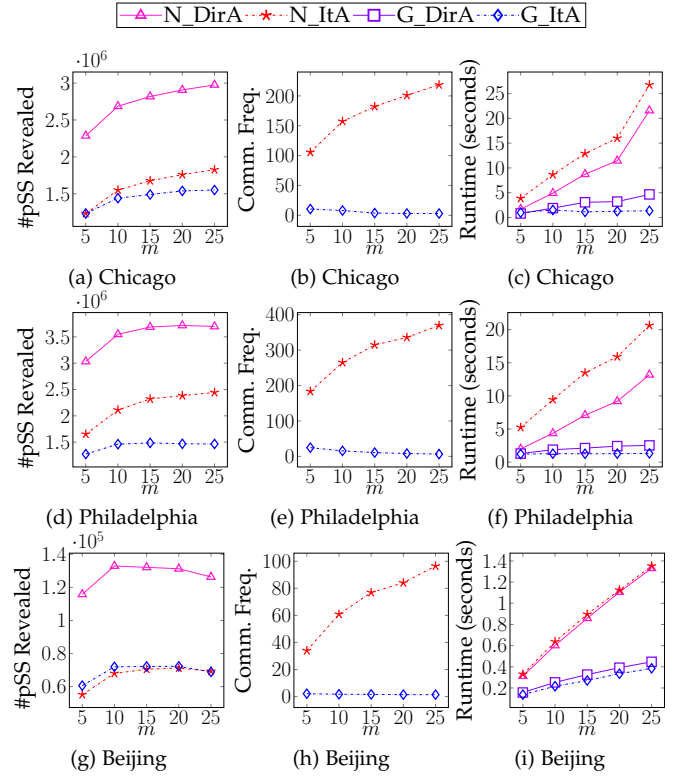

    \centering
    \pgfplotsset{
    compat=newest,
    /pgfplots/legend image code/.code={%
        \draw[mark repeat=3,mark phase=3,#1] 
            plot coordinates {
                (0cm,0cm) 
                (0.4cm,0cm)
                (0.8cm,0cm)
                (1.2cm,0cm)
                (1.6cm,0cm)%
            };
    },
}

\scalebox{0.35}{
\begin{tikzpicture} 
    \begin{axis}[
        hide axis,
        xmin=0,
        xmax=0,
        ymin=-1,
        ymax=-1,
        legend style={at={(1.5,1.15)},anchor=south,legend columns=4,font=\fontsize{24}{0}\selectfont,},
        mark options={mark size=6pt, line width=2pt},
        cycle list name=my four fsr colors,
        ]
        \addplot {0};
        \addlegendentry{\ndir{}};
        \addplot {0};
        \addlegendentry{\nitr{}};
        \addplot {0};
        \addlegendentry{\gdir{}};
        \addplot {0};
        \addlegendentry{\gitr{}};
    \end{axis}
\end{tikzpicture}
}
    \input{safepath/plots/fsr-nPOIs-chi}
    \input{safepath/plots/fsr-nPOIs-phl}
    \input{safepath/plots/fsr-nPOIs-small-beijing}
    \caption{The effect of varying $m$ on the performance metrics for FSR queries (all datasets)}
    \label{fig:exp:fsr:m}
\end{figure}
\begin{figure}[tbh!]
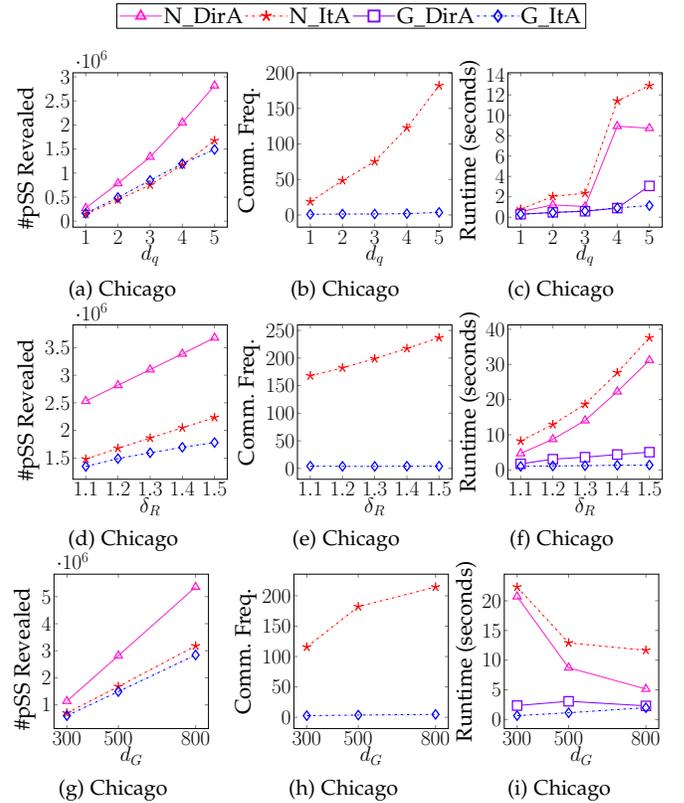

    \centering
    \pgfplotsset{
    compat=newest,
    /pgfplots/legend image code/.code={%
        \draw[mark repeat=3,mark phase=3,#1] 
            plot coordinates {
                (0cm,0cm) 
                (0.4cm,0cm)
                (0.8cm,0cm)
                (1.2cm,0cm)
                (1.6cm,0cm)%
            };
    },
}

\scalebox{0.35}{
\begin{tikzpicture} 
    \begin{axis}[
        hide axis,
        xmin=0,
        xmax=0,
        ymin=-1,
        ymax=-1,
        legend style={at={(1.5,1.15)},anchor=south,legend columns=4,font=\fontsize{24}{0}\selectfont,},
        mark options={mark size=6pt, line width=2pt},
        cycle list name=my four fsr colors,
        ]
        \addplot {0};
        \addlegendentry{\ndir{}};
        \addplot {0};
        \addlegendentry{\nitr{}};
        \addplot {0};
        \addlegendentry{\gdir{}};
        \addplot {0};
        \addlegendentry{\gitr{}};
    \end{axis}
\end{tikzpicture}
}
    \input{safepath/plots/fsr-eDist}
    \input{safepath/plots/fsr-cDist}
    \input{safepath/plots/fsr-xy}
    \caption{The effect of varying $d_q$, $\delta_R$ and $d_G$ on the performance metrics for FSR queries (Chicago dataset)}
    \label{fig:exp:fsr:chi}
\end{figure}

\oc{To compare the performance of our efficient algorithms with the na\"ive ones for FSR queries, we vary $m$, $\delta_R$, $d_q$ and $d_G$. Since the query relevant area of \ndir{} and \gdir{} are the same, they reveal the same number of pSSs and Fig.~\ref{fig:exp:fsr:m} and \ref{fig:exp:fsr:chi} only show the value for \ndir{}. In addition, we omit the communication frequency of \ndir{} and \gitr{} in those plots, since the value is 1.0 for both of them. We follow the same conventions while plotting results in Sections~\ref{exp:gsr} and \ref{exp:gfsr}.}

\oc{Fig.~\subref*{plt:fsr-pp-nPOIs-chi}, \subref*{plt:fsr-pp-nPOIs-phl}, and \subref*{plt:fsr-pp-nPOIs-small-beijing} show that the increasing number of destinations, $m$, does not always increase the number of revealed pSSs. The reason is that when $m$ increases but $d_q$ remains unchanged, nearby safer options may become available which can cause the length of the SR to decrease. Therefore, the impact of varying $m$ on the number of revealed pSSs depends on the dataset. On the other hand, Fig.~\subref*{plt:fsr-pp-eDist-chi}, \subref*{plt:fsr-pp-cDist-chi}, and \subref*{plt:fsr-pp-xy-chi} show that the increasing the value of $d_q$, $\delta_R$ or $d_G$ increases the number of revealed pSSs for all algorithms for the Chicago dataset. We omit the graphs for Beijing and Philadelphia datasets as they show similar trends. For FSR queries, the number of pSSs revealed by \nitr{} and \gitr{} are, on average, 52.9\% and 47.3\% of those of direct algorithms. On average, \gitr{} reveals 8\% less pSSs than \nitr{}. Thus, \gitr{} preserves privacy better for FSR queries. }

\oc{Increasing the value of $m$ causes rapid communication frequency increase for \nitr{}, whereas for \gitr{}, it decreases very slowly (Fig.~\subref*{plt:fsr-cf-nPOIs-chi},~\subref*{plt:fsr-cf-nPOIs-phl}, and \subref*{plt:fsr-cf-nPOIs-small-beijing}). On the other hand, increasing the value of $d_q$ or $\delta_R$ increases the communication frequency for both \nitr{} and \gitr{}. However, the rate of change is sharp for \nitr{} and quite slow for \gitr{}. Fig.~\subref*{plt:fsr-cf-eDist-chi} and \subref*{plt:fsr-cf-cDist-chi} show this trend for the Chicago dataset. Philadelphia and Beijing datasets show similar trends (not shown). 
\nc{The comm. freq. shows an irregular trend for varying $d_G$ because the change in group size does not follow any regular pattern. }
Note that the average communication frequency of \nitr{} and \gitr{} is 157.7 and 4.6, respectively, and, \gitr{} decreases on average 80\% of the communication frequency of \nitr{}.}

\oc{Increasing the value of  $m$, $d_q$, or $\delta_R$ increases the runtime for all four algorithms for all datasets. In addition, the rate of increase of the runtime is higher for the na\"ive algorithms. For example, Fig. \subref*{plt:fsr-rt-nPOIs-chi}, \subref*{plt:fsr-rt-nPOIs-phl} and \subref*{plt:fsr-rt-nPOIs-small-beijing} show the results for all datasets for varying $m$. 
\nc{ With the increase of $d_G$, the search space does not change, but other factors (e.g., our pruning techniques and communication overhead) affect runtime in complex ways and result in an irregular runtime.}
For FSR queries, the average runtimes of \ndir{}, \nitr{}, \gdir{}, and \gitr{} are 6.2, 9.14, 1.6, and 0.8 seconds, respectively. Therefore, our efficient algorithms are 4 to 11 times faster than the na\"ive algorithms.}

\subsubsection{GSR Queries}\label{exp:gsr}

\begin{figure}[tbh!]
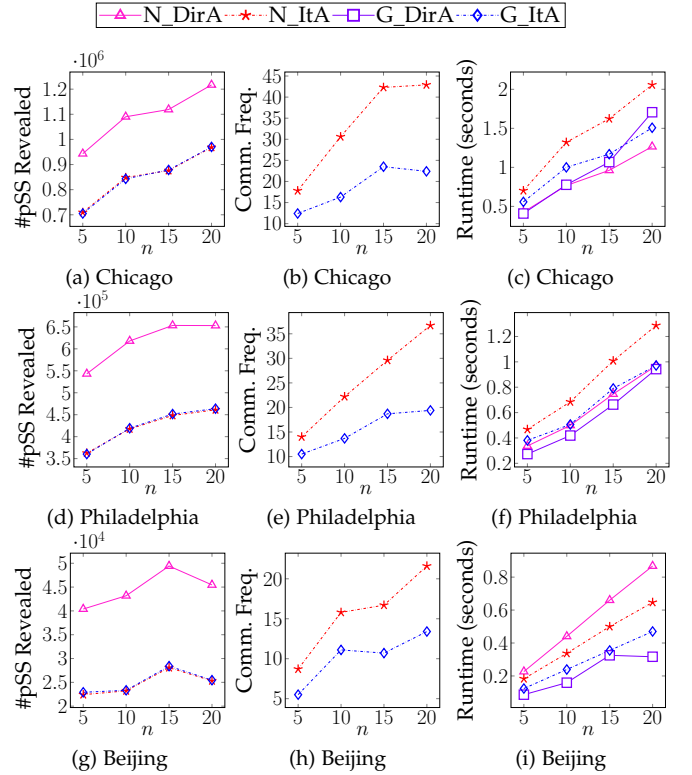

    \centering
    \pgfplotsset{
    compat=newest,
    /pgfplots/legend image code/.code={%
        \draw[mark repeat=3,mark phase=3,#1] 
            plot coordinates {
                (0cm,0cm) 
                (0.4cm,0cm)
                (0.8cm,0cm)
                (1.2cm,0cm)
                (1.6cm,0cm)%
            };
    },
}
\centerline{\begin{minipage}[t]{\columnwidth}
\centering
\scalebox{0.35}{
\begin{tikzpicture} 
    \begin{axis}[
        hide axis,
        xmin=0,
        xmax=0,
        ymin=-1,
        ymax=-1,
        legend style={at={(0.5,1.15)},anchor=south,legend columns=4,font=\fontsize{24}{0}\selectfont,},
        mark options={mark size=6pt, line width=2pt},
        cycle list name=my four fsr colors,
        ]
        \addplot {0};
        \addlegendentry{\ndir{}};
        \addplot {0};
        \addlegendentry{\nitr{}};
        \addplot {0};
        \addlegendentry{\gdir{}};
        \addplot {0};
        \addlegendentry{\gitr{}};
    \end{axis}
\end{tikzpicture}
}
\end{minipage}
}
    \input{safepath/plots/gsr-nG-chi}
    \input{safepath/plots/gsr-nG-phl}
    \input{safepath/plots/gsr-nG-small-beijing}
    \caption{The effect of varying $n$ on the performance metrics for GSR queries (all datasets)}
    \label{fig:exp:gsr:n}
\end{figure}
\begin{figure}[tbh!]
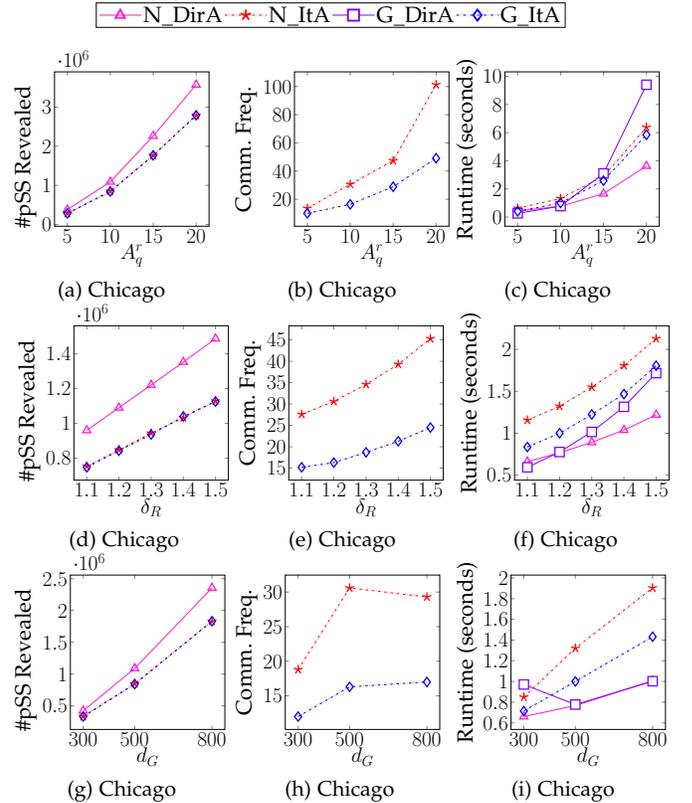

    \centering
    \pgfplotsset{
    compat=newest,
    /pgfplots/legend image code/.code={%
        \draw[mark repeat=3,mark phase=3,#1] 
            plot coordinates {
                (0cm,0cm) 
                (0.4cm,0cm)
                (0.8cm,0cm)
                (1.2cm,0cm)
                (1.6cm,0cm)%
            };
    },
}
\centerline{\begin{minipage}[t]{\columnwidth}
\centering
\scalebox{0.35}{
\begin{tikzpicture} 
    \begin{axis}[
        hide axis,
        xmin=0,
        xmax=0,
        ymin=-1,
        ymax=-1,
        legend style={at={(0.5,1.15)},anchor=south,legend columns=4,font=\fontsize{24}{0}\selectfont,},
        mark options={mark size=6pt, line width=2pt},
        cycle list name=my four fsr colors,
        ]
        \addplot {0};
        \addlegendentry{\ndir{}};
        \addplot {0};
        \addlegendentry{\nitr{}};
        \addplot {0};
        \addlegendentry{\gdir{}};
        \addplot {0};
        \addlegendentry{\gitr{}};
    \end{axis}
\end{tikzpicture}
}
\end{minipage}
}
    \input{safepath/plots/gsr-qA-chi}
    \input{safepath/plots/gsr-cDist}
    \input{safepath/plots/gsr-xy}
    \caption{The effect of varying $A_q^r$, $\delta_R$, and $d_G$ on the performance metrics for GSR queries (Chicago dataset)}
    \label{fig:exp:gsr:chi}
\end{figure}
\oc{We compare the impact of varying $n$, $A_q^r$, $\delta_R$, and $d_G$, for the efficient and na\"ive algorithms for GSR queries.}

\oc{Fig.~\subref*{plt:gsr-pp-nG-chi}, \subref*{plt:gsr-pp-nG-phl}, and \subref*{plt:gsr-pp-nG-small-beijing} show that increasing the group size, $n$, does not necessarily increase the number revealed pSSs in all datasets. On the other hand, increasing $\delta_R$, $d_G$ or $A_q^r$ increases the number revealed pSSs for GSR queries. Fig.~\subref*{plt:gsr-pp-queryArea-chi}, \subref*{plt:gsr-pp-cDist-chi}, and \subref*{plt:gsr-pp-xy-chi} show this trend for the Chicago dataset. Philadelphia and Beijing datasets show similar trends (not shown).  Here, \nitr{} and \gitr{} reveal on average 56.16\% and 56.24\% of the pSSs revealed by the direct algorithms, respectively. Thus, \gitr{} reveals slightly more (0.2\%) pSSs than \nitr{}. For $X_{it}=1$, \nitr{} always reveals more pSSs than \gitr{}. However, as $X_{it}$ increases, more pSSs are revealed for \gitr{} which eventually exceeds \nitr{} (not shown). The reason is that when the frontier of the search in \gitr{} is near a destination, due to the high value of $X_{it}$, it retrieves some edges' pSSs that are not necessary to reach the destination. Though \nitr{} does the same, the number is less for it as the first pruning criteria here prunes slightly more unnecessary edges in this situation.}

\oc{Fig.~\subref*{plt:gsr-cf-nG-chi}, \subref*{plt:gsr-cf-nG-phl} and \subref*{plt:gsr-cf-nG-small-beijing} show that the impact of $n$ on the communication frequency depends on the dataset. Nevertheless, the communication frequency increases with the increasing value of $A_q^r$ or $\delta_R$ for all datasets. Fig~\subref*{plt:gsr-cf-queryArea-chi} and  \subref*{plt:gsr-cf-cDist-chi} show this trend for the Chicago dataset. Philadelphia and Beijing datasets show the same trends (not shown). 
\nc{The effect of $d_G$ on comm. freq. for GSR queries is similar to that of FSR queries. }
Note that for GSR queries, the communication frequency of \nitr{} and \gitr{} are, on average, 26.6 and 15.6, respectively. Hence, \gitr{} decreases the communication frequency significantly (22.2\%) compared to \nitr{} by slightly compromising privacy. }

\oc{The runtime increases with the increasing values of $n$, $A_q^r$, and $\delta_R$ for all algorithms for all datasets. 
\nc{The effect of $d_G$ on runtime for GFSR queries is similar to that of FSR queries. }
The average runtimes of \ndir{}, \nitr{}, \gdir{}, and \gitr{} are 0.8, 1.1, 0.9, and 0.8 seconds, respectively, for GSR queries.}

\begin{figure}[tbh!]
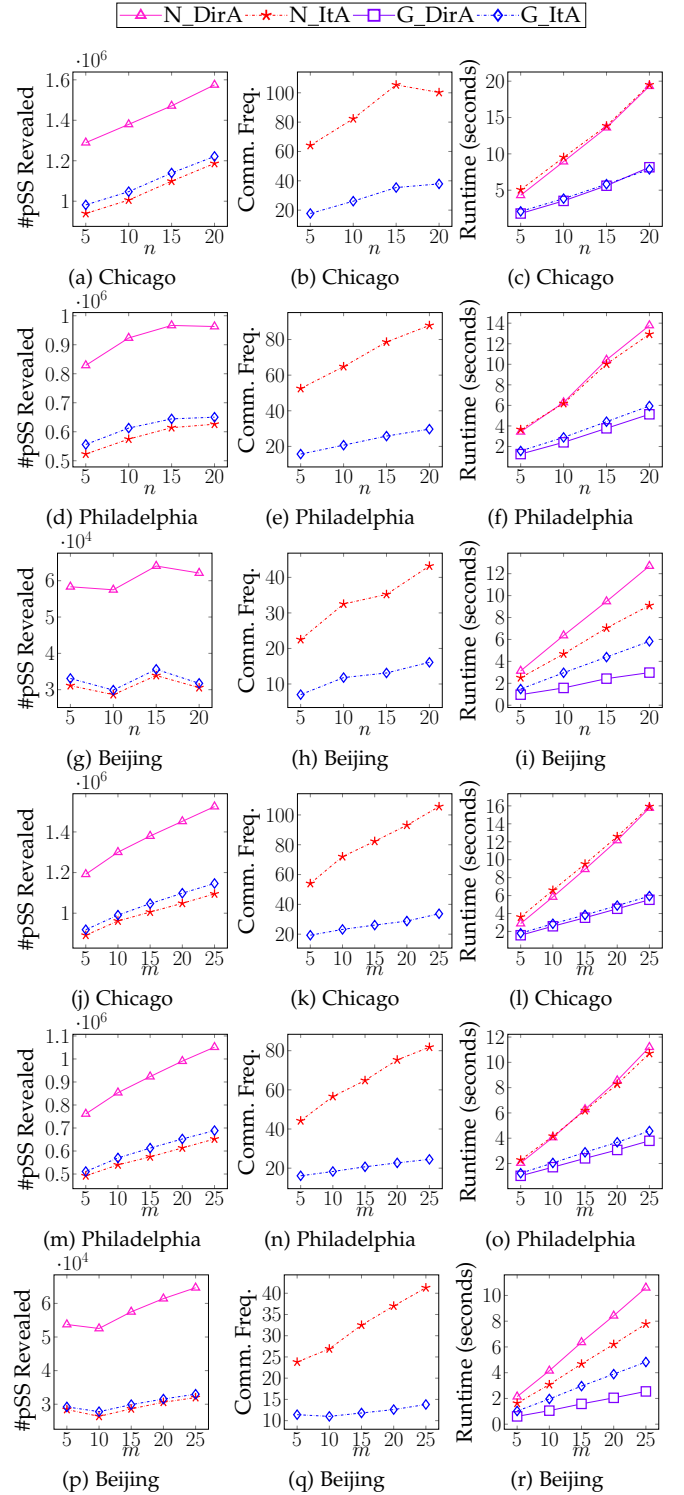

    \centering
    \pgfplotsset{
    compat=newest,
    /pgfplots/legend image code/.code={%
        \draw[mark repeat=3,mark phase=3,#1] 
            plot coordinates {
                (0cm,0cm) 
                (0.4cm,0cm)
                (0.8cm,0cm)
                (1.2cm,0cm)
                (1.6cm,0cm)%
            };
    },
}

\scalebox{0.35}{
\begin{tikzpicture} 
    \begin{axis}[
        hide axis,
        xmin=0,
        xmax=0,
        ymin=-1,
        ymax=-1,
        legend style={at={(1.5,1.15)},anchor=south,legend columns=4,font=\fontsize{24}{0}\selectfont,},
        mark options={mark size=6pt, line width=2pt},
        cycle list name=my four fsr colors,
        ]
        \addplot {0};
        \addlegendentry{\ndir{}};
        \addplot {0};
        \addlegendentry{\nitr{}};
        \addplot {0};
        \addlegendentry{\gdir{}};
        \addplot {0};
        \addlegendentry{\gitr{}};
    \end{axis}
\end{tikzpicture}
}
    \input{safepath/plots/gfsr-nG-chi}
    \input{safepath/plots/gfsr-nG-phl}
    \input{safepath/plots/gfsr-nG-small-beijing}
    \input{safepath/plots/gfsr-nPOIs-chi}
    \input{safepath/plots/gfsr-nPOIs-phl}
    \input{safepath/plots/gfsr-nPOIs-small-beijing}
    \caption{The effect of varying $n$ and $m$ on the performance metrics for GFSR queries (all datasets)}
    \label{fig:exp:gfsr:nm}
\end{figure}
\begin{figure}[tbh!]
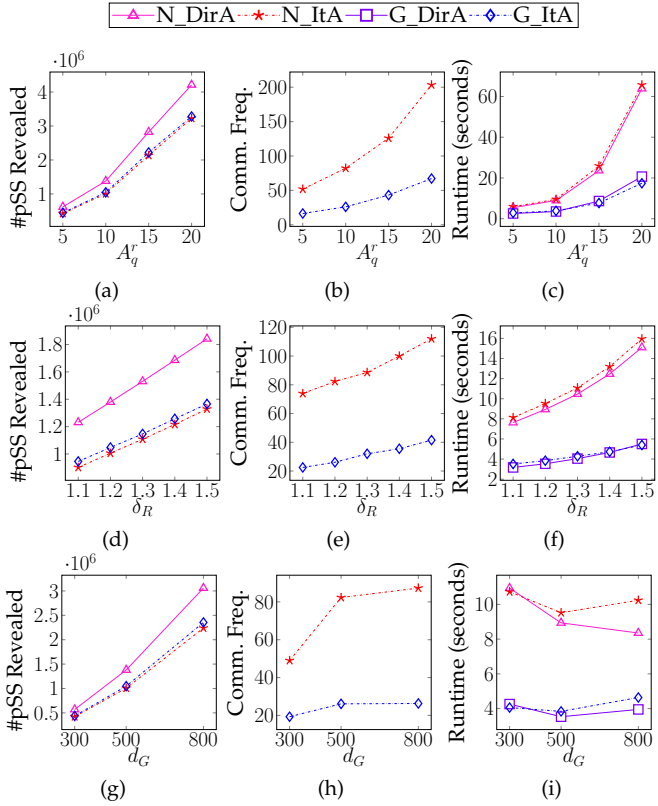

    \centering
    \pgfplotsset{
    compat=newest,
    /pgfplots/legend image code/.code={%
        \draw[mark repeat=3,mark phase=3,#1] 
            plot coordinates {
                (0cm,0cm) 
                (0.4cm,0cm)
                (0.8cm,0cm)
                (1.2cm,0cm)
                (1.6cm,0cm)%
            };
    },
}

\scalebox{0.35}{
\begin{tikzpicture} 
    \begin{axis}[
        hide axis,
        xmin=0,
        xmax=0,
        ymin=-1,
        ymax=-1,
        legend style={at={(1.5,1.15)},anchor=south,legend columns=4,font=\fontsize{24}{0}\selectfont,},
        mark options={mark size=6pt, line width=2pt},
        cycle list name=my four fsr colors,
        ]
        \addplot {0};
        \addlegendentry{\ndir{}};
        \addplot {0};
        \addlegendentry{\nitr{}};
        \addplot {0};
        \addlegendentry{\gdir{}};
        \addplot {0};
        \addlegendentry{\gitr{}};
    \end{axis}
\end{tikzpicture}
}
    \input{safepath/plots/gfsr-qA}
    \input{safepath/plots/gfsr-cDist}
    \input{safepath/plots/gfsr-xy}
    \caption{The effect of varying $A_q^r$, $\delta_R$ and $d_G$ on the performance metrics for GFSR queries (Chicago dataset)}
    \label{fig:exp:gfsr:chi}
\end{figure}
\subsubsection{GFSR Queries} \label{exp:gfsr}
\oc{We vary $n$, $m$, $A_q^r$, $\delta_R$, and $d_G$ to compare our safe route planner with the na\"ive algorithms for GFSR queries.}

\oc{For GFSR queries, the effect of increasing $n$ or $m$ on privacy (i.e., the number revealed pSSs) depends on the dataset (Fig.~\subref*{plt:gfsr-pp-nG-chi}, \subref*{plt:gfsr-pp-nG-phl}, \subref*{plt:gfsr-pp-nG-small-beijing}, \subref*{plt:gfsr-pp-nPOIs-chi}, \subref*{plt:gfsr-pp-nPOIs-phl}, \subref*{plt:gfsr-pp-nPOIs-small-beijing}). Contrarily, increasing the values of $A_q^r$, $\delta_R$ or $d_G$ increases the number of pSSs disclosure for all datasets. Fig. \subref*{plt:gfsr-pp-queryArea-chi}, \subref*{plt:gfsr-pp-cDist-chi} and \subref*{plt:gfsr-pp-xy-chi} show this trend for Chicago dataset. For GFSR queries, \nitr{} and \gitr{} reveal, on average, 52.6\% and 54.9\% of the pSSs revealed by the direct algorithms. Hence, \gitr{} reveals slightly more pSSs (3.8\%) than \nitr{} for GFSR queries. }

\oc{The communication frequency of \nitr{} and \gitr{} increases when increasing $n$ or $m$ for all datasets (Fig. \subref*{plt:gfsr-cf-nG-chi}, \subref*{plt:gfsr-cf-nG-phl}, \subref*{plt:gfsr-cf-nG-small-beijing}, \subref*{plt:gfsr-cf-nPOIs-chi}, \subref*{plt:gfsr-cf-nPOIs-phl}, and \subref*{plt:gfsr-cf-nPOIs-small-beijing}). It also increases with the increasing values of $A_q^r$ , $\delta_R$ for all datasets. Fig. \subref*{plt:gfsr-cf-queryArea-chi}, \subref*{plt:gfsr-cf-cDist-chi} show this trend for the Chicago dataset. 
\nc{The effect of $d_G$ here is similar to that of FSR queries.}
The average communication frequencies of \nitr{} and \gitr{} are 65.2 and 22.1, respectively. Therefore, \gitr{} finds the answer of GFSR queries with 48\% less communication frequency while compromising privacy slightly.}

\oc{Increasing the values of $n$ or $m$ increases the runtime for all algorithms for all datasets (Fig. \ref{fig:exp:gfsr:nm}). The runtime also increases for both na\"ive and efficient algorithms when $A_q^r$ or $\delta_R$ is increased in all datasets. Fig. \subref*{plt:gfsr-rt-queryArea-chi} and \subref*{plt:gfsr-rt-cDist-chi} show this trend for the Chicago dataset. The trends for Philadelphia and Beijing are the same for varying $A_q^r$ or $\delta_R$ (not shown). 
\nc{The effect of $d_G$ here is similar to that of FSR queries (e.g. Fig. \subref*{plt:gfsr-rt-xy-chi}). } 
The average runtime of \ndir{}, \nitr{}, \gdir{}, and \gitr{} are 9.6, 9.2, 3.3, and 3.9 seconds, respectively.}

\subsubsection{Modified $R$-tree} \label{exp:rtree}

\nc{Our system updates pSSs in the modified $R$-tree daily based on the visited route of the users. The update time includes the time for insert, search and delete operations. Table~\ref{tab:rtree} shows the average time to insert/delete/update pSSs daily per user, along with the average number of pSSs considered for each operation. Table~\ref{tab:rtree} also shows the average pSS search time and the average number of pSSs retrieved per 100 queries for \dir{} and \itr{} in the default setting. All types of operations incur very low processing overhead, which is practically acceptable.  Moreover, a user stores, on average, 10600 pSSs using 4400 supercells. Storing the pSSs by grouping them into supercells saves, on average, 53\% of user storage. 
}



\begin{table}[!htb]
\caption{\nc{Runtime of our $R$-tree operations}}
\label{tab:rtree}
\vspace{-2mm}
\centering
\nc{
    \begin{tabular}{|L{1.2cm}|R{0.8cm}|r|R{0.8cm}|r|R{0.8cm}|r|}
    \hline
        \multirow{2}{*}{Operation} & \multicolumn{2}{c|}{Chicago} & \multicolumn{2}{c|}{Philadelphia} & \multicolumn{2}{c|}{Beijing} \\ \cline{2-7}
            & Time (ms) & \#pSSs & Time (ms) & \#pSSs & Time (ms) & \#pSSs \\ \hline\hline
        Insert          & 14.444 & 3736 
                        & 16.867 & 6991 
                        & 19.194 & 5552 \\ \hline 
        Delete          & 0.898 & 1549 
                        & 1.196 & 2709 
                        & 1.625 & 4015 \\ \hline
        Update          & 19.156 & 4466 
                        & 31.270 & 9116
                        & 40.214 & 8621\\ \hline 
        search (\dir{}) & 0.195 & 3473 
                        & 0.397 & 7289 
                        & 0.184 & 3075 \\ \hline
        search (\itr{}) & 0.060 & 387  
                        & 0.060 & 279
                        & 0.055 & 553 \\ \hline
    \end{tabular}
}
\vspace{-5mm}
\end{table}




\subsection{Impact of Missing Data} \label{exp:missing}
A centralized architecture assumes that users share their travel experiences with a centralized server without considering privacy issues. However, in reality, this does not happen and the centralized solution has missing data. We investigate the impact of missing data on the quality of SRs.

As mentioned before, there exists no solution for finding SRs in our problem setting. Thus, for this experiment, we adopt our solution for the centralized model, where users share pSSs with the centralized server. We compare the accuracy and confidence level of our system with the centralized architecture. We vary the percentage of available data for the centralized model as 50\%, 60\%, 70\%, 80\%, and 90\% and denote them with C50, C60, C70, C80, and C90, respectively.

\subsubsection{Accuracy.}

\begin{figure}[!htb]
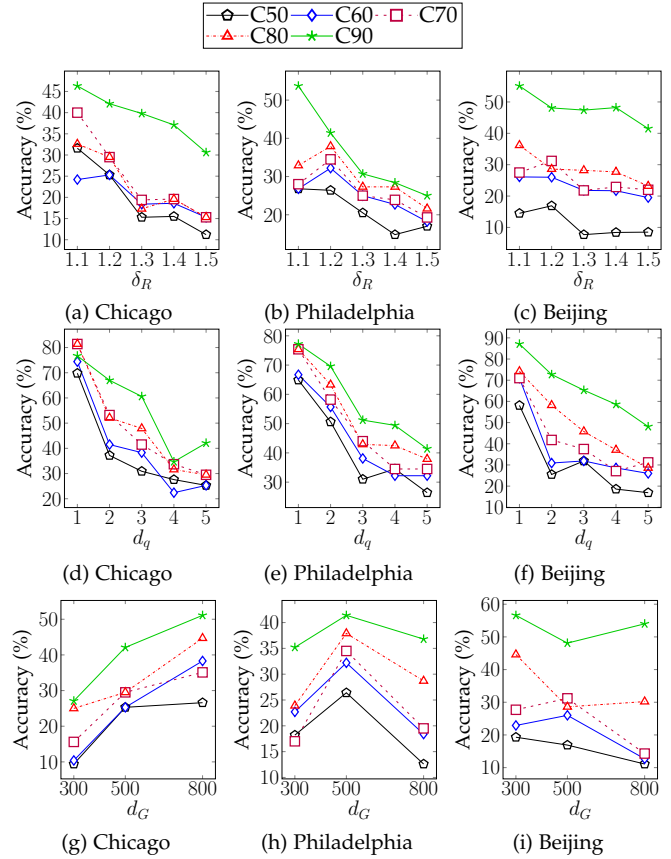

    \centering
    \pgfplotsset{
    compat=newest,
    /pgfplots/legend image code/.code={%
        \draw[mark repeat=3,mark phase=3,#1] 
            plot coordinates {
                (0cm,0cm) 
                (0.4cm,0cm)
                (0.8cm,0cm)
                (1.2cm,0cm)
                (1.6cm,0cm)%
            };
    },
}
\centerline{\begin{minipage}[t]{\columnwidth}
\centering
\scalebox{0.35}{
\begin{tikzpicture} 
    \begin{axis}[
        hide axis,
        xmin=0,
        xmax=0,
        ymin=-1,
        ymax=-1,
        legend style={at={(0.5,1.15)},anchor=south,legend columns=3,font=\fontsize{24}{0}\selectfont,},
        mark options={mark size=6pt, line width=2pt},
        cycle list name=my accuracy five colors,
        ]
        \addplot {0};
        \addlegendentry{C50};
        \addplot {0};
        \addlegendentry{C60};
        \addplot {0};
        \addlegendentry{C70};
        \addplot {0};
        \addlegendentry{C80};
        \addplot {0};
        \addlegendentry{C90};
    \end{axis}
\end{tikzpicture}
}
\end{minipage}
}
    \input{safepath/plots/rank-cDist}
    \vspace{0.1cm}
    \input{safepath/plots/rank-eDist}
    \vspace{0.1cm}
    \input{safepath/plots/rank-xy}
    \caption{Accuracy loss in the centralized model for  missing data for SR queries. C50 means 50\% of actual data is present.}
    \label{fig:plt:accuracy-centralized}
\end{figure}

In our system, users do not hesitate to share their pSSs as there is no fear of privacy violation. Thus, our system always provides the actual SR. We measure the accuracy as the percentage of the answers that are within the top-5 SRs. Fig. \ref{fig:plt:accuracy-centralized} shows that the average accuracy increases with an increase in user data (25.4\% for C50 and 49.9\% for C90) \oc{for SR queries}. Even 10\% missing data causes significant (50.1\%) accuracy loss. Hence it is important to adopt privacy preserving solution to find the SRs. For the same amount of available data, the accuracy decreases with the increase of $\delta$ and $d_q$, because the number of possible routes from $s$ to $d$ increases. The accuracy does not depend on $d_G$ (Fig. \subref*{plt:rank-xy-chi}-\subref*{plt:rank-xy-beijing}). 

\oc{For FSR, GSR and GFSR queries, we compute the accuracy on the same percentages of missing data in the default setting for all datasets (Table \ref{tab:missing:acc}). Here, the accuracy is measured as the percentage of the answers that are within the top-5 SRs for FSR queries and within the top-5 safest route sets for GSR and GFSR queries. For all of those queries, the accuracy keeps decreasing with the increase of missing data. For FSR queries, 10\% missing data causes the accuracy to drop to, on average, 37\%, and it drops further to 9.5\% for 50\% missing data for \nc{the centralized model}. Similarly, for GSR and GFSR queries 10\% missing data causes the accuracy to drop to 49.3\% and 42.4\%, respectively, \nc{for the centralized model}.}

\begin{table}[tbh!]
    \caption{Accuracy loss in the centralized model for missing data for FSR, GSR and GFSR queries in the default setting.}
    \label{tab:missing:acc}
    \centering
    \nc{
    \begin{tabular}{|C{0.9cm}|c|r|r|r|r|r|}
    \hline
     & & \multicolumn{5}{|c|}{Accuracy (\%)} \\ \cline{3-7}
    \multirow{-2}{0.9cm}{Query-type} & \multirow{-2}{*}{Dataset} & C50 & C60 & C70 & C80 & C90 \\ \hline \hline   
    \multirow{3}{*}{FSR} 
    & C & 8.4  & 6.3 & 13.7 & 16.8 & 25.3 \\ \cline{2-7}
    & P & 14.0 & 15.1 & 22.1 & 18.6 & 34.9 \\ \cline{2-7}
    & B & 6.0 &  7.5 & 11.9 & 23.9 & 50.7 \\ \hline\hline
    \multirow{3}{*}{GSR}
    & C & 22.6 & 29.0 & 32.3 & 43.0 & 41.9 \\ \cline{2-7}
    & P & 36.7 & 35.4 & 44.3 & 51.9 & 58.2 \\ \cline{2-7}
    & B & 28.3 & 23.9 & 26.1 & 34.8 & 47.8 \\ \hline\hline
    \multirow{3}{*}{GFSR}
    & C & 17.0 & 17.0 & 25.5 & 27.7 & 39.4\\ \cline{2-7}
    & P & 18.5 & 24.7 & 35.8 & 34.6 & 45.7\\ \cline{2-7}
    & B &  4.0 & 12.0 & 10.0 & 26.0 & 42.0\\ \hline
    \end{tabular}
    }
\end{table}

\subsubsection{Confidence Level (\textbf{CL})} \label{confidence-level}

\begin{figure}[bt!]
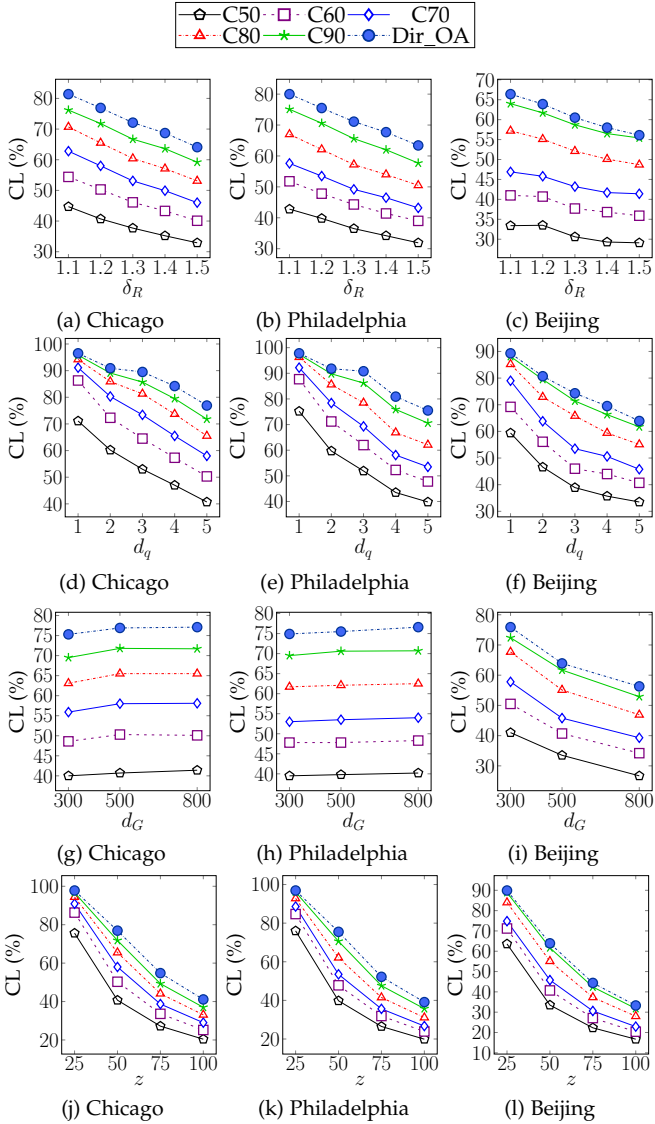

    \centering
    \pgfplotsset{
    compat=newest,
    /pgfplots/legend image code/.code={%
        \draw[mark repeat=3,mark phase=3,#1] 
            plot coordinates {
                (0cm,0cm) 
                (0.4cm,0cm)
                (0.8cm,0cm)
                (1.2cm,0cm)
                (1.6cm,0cm)%
            };
    },
}
\centerline{\begin{minipage}[t]{\columnwidth}
\centering
\scalebox{0.35}{
\begin{tikzpicture} 
    \begin{axis}[
        hide axis,
        xmin=0,
        xmax=0,
        ymin=-1,
        ymax=-1,
        legend style={at={(0.5,1.15)},anchor=south,legend columns=3,font=\fontsize{24}{0}\selectfont,},
        mark options={mark size=6pt, line width=2pt},
        cycle list name=my cl six colors,
        ]
        \addplot {0};
        \addlegendentry{C50};
        \addplot {0};
        \addlegendentry{C60};
        \addplot {0};
        \addlegendentry{C70};
        \addplot {0};
        \addlegendentry{C80};
        \addplot {0};
        \addlegendentry{C90};
        \addplot {0};
        \addlegendentry{Dir\_OA};
    \end{axis}
\end{tikzpicture}
}
\end{minipage}
}
    \input{safepath/plots/cl-cDist}
    \vspace{1mm}
    \input{safepath/plots/cl-eDist}
    \vspace{0.1cm}
    \input{safepath/plots/cl-xy}
    \vspace{0.1cm}
    \input{safepath/plots/cl-z}
    \caption{{CL for our system is higher than that of  the centralized model (SR queries).}}
    \label{fig:plt:confidence-level}
\end{figure}

Fig. \ref{fig:plt:confidence-level} shows that the CL for our system is always the highest {(on average 75.7\%)} for SR queries. Since both \dir{} and \itr{} provide optimal solutions, their CL is the same. In the centralized model, CL predictably increases with the increase of missing data. CL decreases when the SRs become longer (for $\delta$ and $d_q$). No particular trend is visible for $d_G$. 
\oc{The CL decreases with the increase in $d_G$ as the SR contains more cells and ensuring on average 50\% (the default value of $z$=50) knowledgeable users per cell becomes more difficult. The parameter $z$ included in the definition of the confidence level as discussed in Section~\ref{sec:CL} is varied within $\{25, \textbf{50}, 75, 100\}$ to cover the full range. For an increase in $z$, the CL decreases as expected.}  

\oc{For FSR, GSR and GFSR queries, we compute the CLs in default setting for all datasets (result not shown). The average CLs of FSR, GSR and GFSR queries are 47.2\%, 79.6\% and 71.7\% for our algorithms, whereas, \nc{for the centralized model,} 10\% missing data causes the CLs to drop to 43.1\%, 75.5\% and 66.8\%; moreover, 50\% missing data drops those values further to 24.3\%, 42.7\% and 36.9\%. }

\subsection{Effectiveness of SRs} \label{exp:query-effect}
\nc{We evaluate the effectiveness of the SRs with two sets of experiments: with respect to (i) the shortest routes and (ii) the SR returned by Dijkstra's algorithm.}

\emph{Comparison with the shortest routes. }
\oc{
Table~\ref{tab:SR-vs-SP} shows the results of this experiment for 
100 SR queries in the default setting. 
For each query, we compute the top-$K$ shortest routes using Yen's algorithm and check if any of those top-$K$ routes that are within $\delta$ are as safe as our SR. 
Table~\ref{tab:SR-vs-SP} shows that only 2.7\% among the shortest routes (K = 1) are the SRs. Increasing $K$ increases the percentage barely. Even for a high value of $K = 500$, only 27\% are the SRs. This is because the consecutive shortest routes (e.g., the third and the fourth ones) normally have very small differences in terms of the included roads in the routes.} 
\oc{
In addition, the computation time of the top-$K$ shortest routes exceeds that of \gdir{} for $K=50$ and keeps increasing for higher values of $K$. Computing the top-500 shortest routes takes, on average, 74 times more computation time than \gdir{}, and yet, the SR is not found in most cases.
}

\begin{table}[!htb]
\caption{The percentage of query samples for which top-$K$ shortest routes (ShR) include the respective SRs and the time needed to calculate them }
\label{tab:SR-vs-SP}
\hspace*{-3mm}
\centering
\begin{tabular}{|l|R{0.028\textwidth}|R{0.048\textwidth}|R{0.042\textwidth}|R{0.048\textwidth}|R{0.042\textwidth}|R{0.048\textwidth}|R{0.042\textwidth}|}
\hline
\multicolumn{2}{|c|}{\multirow{2}{*}{Algo}} & \multicolumn{2}{c|}{Chicago} & \multicolumn{2}{c|}{Philadelphia} & \multicolumn{2}{c|}{Beijing} \\ \cline{3-8}

\multicolumn{2}{|c|}{} & (\%) SRs & Time (sec.) & (\%) SRs & Time (sec.) & (\%) SRs & Time (sec.) \\ \hline\hline

\multicolumn{2}{|c|}{\dir{}} & \textbf{100.0} & \textbf{0.3} & \textbf{100.0} & \textbf{0.01} & \textbf{100.0} & \textbf{0.5} \\ \hline

\multirow{6}{*}{\rotatebox[origin=c]{90}{Top-$K$ ShR}} & 1 & 1.0 & 0.0 & 4.8 & 0.0 & 2.3 & 0.0 \\ \cline{2-8} 
 & 10 & 3.1 & 0.2 & 2.3 & 0.4 & 11.9 & 0.1 \\ \cline{2-8} 
 & 50 & 3.1 & 1.3 & 2.3 & 2.4 & 13.1 & 0.9 \\ \cline{2-8} 
 & 100 & 3.1 & 2.8 & 2.3 & 5.0 & 19.1 & 1.8 \\ \cline{2-8} 
 & 250 & 3.1 & 7.8 & 2.3 & 14.0 & 25.0 & 5.0 \\ \cline{2-8} 
 & 500 & 3.1 & 18.2 & 3.4 & 33.1 & 27.4 & 11.7 \\ \hline
\end{tabular}
\end{table}


\begin{table}[!htb]
\caption{\nc{The length of the SR returned by \gdir{} and Dijkstra's algorithm  under the default setting, the length ratio of the SR to the shortest path, and the time needed to calculate them}}
\label{tab:SR:dir-vs-dijkstra}
\vspace{-2mm}
\hspace*{-2mm}
\centering
\nc{
    \begin{tabular}{|L{0.041\textwidth}
    |R{0.03\textwidth}|R{0.025\textwidth}|R{0.024\textwidth}
    |R{0.03\textwidth}|R{0.025\textwidth}|R{0.024\textwidth}
    |R{0.03\textwidth}|R{0.025\textwidth}|R{0.024\textwidth}|}
    \hline
        \multirow{2}{*}{Algo} & \multicolumn{3}{c|}{Chicago} & \multicolumn{3}{c|}{Philadelphia} & \multicolumn{3}{c|}{Beijing} \\ \cline{2-10}
            & \barr Len\\(m) \earr & Ratio & \barr Time\\(sec.) \earr 
            & \barr Len\\(m) \earr & Ratio & \barr Time\\(sec.) \earr 
            & \barr Len\\(m) \earr & Ratio & \barr Time\\(sec.) \earr \\\hline\hline
        Shortest & 6958 & 1.00 & 0.0  
                & 6846 & 1.00 & 0.0 
                & 7815 & 1.00 & 0.0\\ \hline
        \gdir{} & 8162 & 1.17 & 0.31      
                & 8082 & 1.18 & 0.42
                & 8978 & 1.15 & 0.06\\ \hline
        Dijkstra (Safest) 
                & 31210 & 4.53 & 0.02 
                & 27852 & 4.24 & 0.02
                & 42008 & 5.68 & 0.02\\ \hline
    \end{tabular}
}
\end{table}

\emph{Comparison with the SR returned by Dijkstra's algorithm. }
\nc{
The SR computed by Dijkstra's algorithm cannot consider any distance constraint. 
For this experiment, we calculate the SRs using Dijkstra for 100 SR queries in the default setting and show the results in Table~\ref{tab:SR:dir-vs-dijkstra}. 
The average length of the SR by Dijkstra's algorithm is 4.5 (Chicago), 4.2 (Philadelphia) and 5.7 (Beijing) times the shortest path length, and can be as high as 15.38 (Chicago), 17.92 (Philadelphia) and 21.5 (Beijing) times in the worst case. In contrast, our SR is always 1.2 times the shortest path length as intended. Therefore, the SR returned by Dijkstra is infeasible for real-life use.}

\section{Conclusion}\label{sec:conclusion} 
We developed a novel journey planner for finding SRs with crowdsourced data and computation. In experiments, we observe that the data scarcity problem can have a significant impact on lowering the quality of SRs. For example, the actual SR is only identified for on average 36\% and 41\% times when a centralized route planner has 30\% and 20\% missing data, respectively. Our privacy-enhanced solution encourages more users to share their data and improves the quality of the SRs.

\oc{In this paper, we have focused on processing an SR query and its variants: FSR, GSR, and  GFSR queries. Our generalized algorithms can find the query answer in seconds; It takes on average 0.5 seconds for an SR query, 1.2 seconds for an FSR query, 0.9 seconds for a GSR query, and 3.6 seconds for a GFSR query. Our iterative query processing algorithm enhances user privacy by not revealing, on average, 47\% 
of the pSSs revealed by the direct query processing algorithm. The direct one is better than the iterative algorithm in terms of processing time and communication frequency. 
}

\nc{In the future, we plan to extend our solution in the following ways: (i) preserve privacy from malicious attackers, (ii) develop algorithms to find safest routes in a centralized way (i.e., users share their pSSs with a centralized server),  (iii) accommodate grid cells of variable sizes based on city dynamics, and (iv) selecting the appropriate event types to define safety in different contexts (e.g., time and weather of the day, travel mode) and learn impact values of various event type (e.g., crime or accident) in different contexts.}   

\textbf{Acknowledgments} This research has been done in Bangladesh University of Engineering and Technology (BUET). Fariha Tabassum Islam is supported by the ICT Division, Bangladesh (56.00.0000.028.33.108.18).

\bibliographystyle{IEEEtran}
\bibliography{ms}

\end{document}